\documentclass[11pt]{article}
\usepackage{enumerate}
\usepackage[OT1]{fontenc}
\usepackage{amsmath,amssymb, amsthm}
\usepackage[square,sort,comma,numbers]{natbib}
\usepackage[usenames]{color}
\usepackage{dsfont}
\usepackage{pgfplots}
\usepackage{multirow}
\usepackage{rotating}
\usepackage{enumerate}
\usepackage{esvect}
\usepackage{multirow}
\usepackage{tikz}
\usetikzlibrary{patterns}
\usetikzlibrary{arrows}

\usepackage{times}
\usepackage{setspace}
\usepackage{amsfonts}       
\usepackage{nicefrac}       
\usepackage{microtype}      
\usepackage{mathtools}
\usepackage{fullpage}

\usepackage{float}

\usepackage{latexsym}
\usepackage{amssymb}
\usepackage{amsmath}
\usepackage{amsthm}
\usepackage{booktabs}
\usepackage[inline,shortlabels]{enumitem}
\usepackage{graphicx}
\usepackage{color}

\usepackage[switch]{lineno}
\usepackage{makecell}

\usepackage{cleveref}

\usepackage{booktabs}
\usepackage{algorithm}
\usepackage[noend]{algorithmic}
\usepackage[switch]{lineno}
\usepackage{amsmath}
\usepackage{graphicx}
\usepackage{amsthm,amssymb}
\usepackage{comment}
\usepackage{url}

\newtheorem{claim}{Claim}
\usepackage{todonotes}

\newcommand{\decprob}[3]{%
	\begin{center}%
		\begin{minipage}{0.9\linewidth}%
			\textsc{#1}\\
			\textbf{Input:} #2\\
			\textbf{Question:} #3
		\end{minipage}%
	\end{center}%
}

\newcommand{\hf}[1]{{\color{cyan}  [\text{Haifeng:} #1]}}

\usepackage{thmtools}
\usepackage{thm-restate}
\usepackage{diagbox}

\DeclareMathOperator*{\argmin}{arg\,min}

\newtheorem{theorem}{Theorem}
\newtheorem{lemma}[theorem]{Lemma}

\newtheorem{proposition}[theorem]{Proposition}

\makeatletter
\def\moverlay{\mathpalette\mov@rlay}
\def\mov@rlay#1#2{\leavevmode\vtop{%
   \baselineskip\z@skip \lineskiplimit-\maxdimen
   \ialign{\hfil$\m@th#1##$\hfil\cr#2\crcr}}}
\newcommand{\charfusion}[3][\mathord]{
    #1{\ifx#1\mathop\vphantom{#2}\fi
        \mathpalette\mov@rlay{#2\cr#3}
      }
    \ifx#1\mathop\expandafter\displaylimits\fi}
\makeatother

\newcommand{\cupdot}{\charfusion[\mathbin]{\cup}{\cdot}}

\if@double@blind@
  \excludecomment{ack}
\else
  \newenvironment{ack}{\section*{Acknowledgements}}{}
\fi

\newcommand{\BibTeX}{B\kern-.05em{\sc i\kern-.025em b}\kern-.08em\TeX}

\title{Escape Sensing Games:\\  Detection-vs-Evasion in Security Applications}
\author{
Niclas Boehmer\footnote{Equal Contribution.} \\ Harvard University\\ \small  nboehmer@g.harvard.edu
\and
Minbiao Han$^*$ \\ University of Chicago \\ \small minbiaohan@uchicago.edu
\and
Haifeng Xu \\ University of Chicago \\ \small haifengxu@uchicago.edu
\and
Milind Tambe \\ Harvard University \\ \small  milind\_tambe@harvard.edu
}

\begin{document}
\maketitle

\begin{abstract}
Traditional game-theoretic research for security applications primarily focuses on the allocation of external protection resources to defend targets. 
    This work puts forward the study of a new class of games centered around strategically \emph{arranging  targets} to protect them against a constrained adversary, with motivations from varied domains such as peacekeeping resource transit and cybersecurity.
    Specifically, we introduce Escape Sensing Games (ESGs). 
    In ESGs, a blue player manages the order in which targets pass through a channel, while her opponent tries to capture the targets using a set of sensors that need some time to recharge after each activation. 
    We present a thorough computational study of ESGs.
    Among others, we show that it is NP-hard to compute best responses and equilibria. 
    Nevertheless, we propose a variety of effective (heuristic) algorithms whose quality we demonstrate in extensive computational experiments.
\end{abstract}

\section{Introduction}
The past decade has witnessed an influential line of    research in AI, particularly multiagent systems (MAS),  that  employs computational game theory to tackle critical challenges in security and public safety applications, ranging from protecting national ports \cite{shieh2012protect} to combating  smuggling \cite{bucarey2017building} and illegal poaching \cite{fang2015introduction} to defending our cyber systems \cite{vaneky2012game}. At the core of almost all of these problems is  to optimize the allocation of (often limited) \emph{external forces}  to protect critical targets. 
In this work, we adopt a similar computational game theory approach but address a fundamentally different type of security challenge that looks to improve security via  \emph{optimizing the arrangement of targets}  in the face of adversaries.  
 This research contributes a novel perspective to game-theoretic security strategies, emphasizing target arrangement as a  defense mechanism against adversaries. 

Specifically, we introduce and study 
\emph{Escape Sensing Games} (ESGs). 
In these games, a \emph{blue} player aims to securely navigate a set of  \emph{targets} through a channel, whereas her opponent, the \emph{red} player, controls a set of \emph{sensors} along the channel and tries to sense (and therefore ``steal'') as many  targets as possible.
This model captures strategic interactions   arising in various domains.   One example is the transportation of peacekeeping resources using a convoy of ships or cars over a fixed route with malicious actors (e.g., pirates  or hostile forces) trying to intercept them \cite{pirate-un} (see \Cref{sec:Model} for details).  
In cybersecurity, the blue player could model a network administrator routing sensitive data packets through a network with an attacker trying to intercept them. 
Our model captures strategic interactions in these settings arising when security measures are either unavailable or have already been allocated and the blue player is only left with scheduling the targets to avoid detection by the attacker.

We study the optimal sequential play in  ESGs, where the blue player first commits to an ordering of targets followed by the red player devising an optimal sensing plan. Herein, sensors' capabilities are limited in two ways. First, each sensor is only capable of sensing certain targets, modeling that detection and interception technologies are not uniformly effective across different targets, due to differing characteristics such as size, speed, or defense mechanisms. Second,  sensors need a certain time to recharge after sensing a target, modeling limits inherent in detection and interception systems, where permanent action is not feasible.

There are certain challenges integral to our model that make the computation of equilibria highly non-trivial. First, the action space of both players has an exponential size, rendering standard solution approaches such as support enumeration computationally infeasible. 
Second, also after the strategies of both players have been fixed, the game evolves in a complex, sequential fashion with targets moving one after each other through the channel.  
Connected to this, third, it turns out that the red player's best response problem of coming up with an optimal sensing plan given a target ordering is already NP-hard. 
Consequently, this paper also contributes to the algorithmic research on computationally challenging games, a fairly unexplored topic outside of combinatorial game theory   \cite{DBLP:books/daglib/0005753,DBLP:conf/mfcs/LodingR03,DBLP:conf/mfcs/Demaine01}.

\subsection{Our Contribution}
We contribute a new perspective to the rich literature on computational game theory for security applications through our study of the previously overlooked problem of target arrangement. 
Specifically, we introduce and analyze Escape Sensing Games with a focus on the target-controlling blue player. 
We demonstrate that solving this game is highly complex, as we prove that it is NP-hard for both players to compute their optimal strategies. 
To nevertheless be able to solve ESGs in practice, we devise algorithms for computing the red player's strategy, which turn out to scale well in our experiments. 
Computing the blue player's strategy and thereby the game's Stackelberg equilibrium turns out to be a much more intricate task. 
Our experiments show that our formulation of the problem as a bilevel program is only capable of solving small instances of the game exactly. 
Motivated by this, we present a heuristic that effectively combines simulated annealing with a greedy heuristic and an Integer Linear Program (ILP) for computing the red player's strategy. 
We demonstrate the quality of our heuristic through extensive experiments.

We further this investigation in \Cref{non-coor} by 
studying a different variant where sensors are decentralized hence each sensor acts independently according to a simple greedy strategy. We show that it remains NP-hard for the blue player to compute its optimal strategy. While in this setting blue's problem admits an ILP formulation, we demonstrate in experiments that it can only solve up to medium-sized instances. 
Addressing this, we present heuristics that perform well in our experiments. We also demonstrate that while sensors usually have some gain from coordination, this gain depends decisively on the instance structure and is oftentimes rather small (below~$20\%$). 

Full proofs of all results and descriptions of additional experiments can be found in our full version \cite{us}.

\section{Escape Sensing Games (ESGs): The Model}\label{sec:Model}
\begin{figure}[h]
\vspace*{-0.3cm}
\centering
\includegraphics[width=0.45\textwidth]{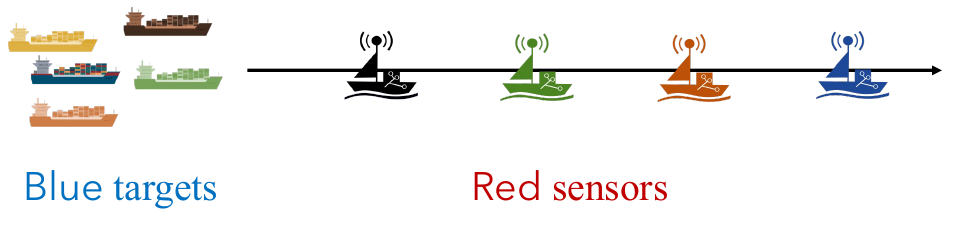}
\vspace{-3mm}
\caption{A visual representation of an Escape Sensing Game.}
\label{fig:model}
\end{figure}

In an \emph{Escape Sensing Game} (ESG)  a blue player (henceforth \texttt{Blue}) tries to route a set of targets through a channel.
They compete against a red player (henceforth \texttt{Red}) who controls a set of sensors and tries to sense (and therefore ``steal'') as many of \texttt{Blue}'s targets as possible.\footnote{Note that the terms   ``sensor'' and ``sensing'' are only part of our terminology and do not limit the applications of our model. For instance, instead of ``sensing'' the targets, \texttt{Red} might also aim to intercept them.}
Formally, an \emph{Escape Sensing Game} is defined by
\begin{enumerate}
    \item a set of targets $T=\{t_1,\dots, t_n\}$, each target $t_i\in T$ equipped with some utility value $v_i\in \mathbb{R}^+$,  
    \item a set of sensors $S=\{s_1,\dots, s_k\}$ and a recharging time $\tau\in \mathbb{N}\cup \{\infty\}$ which is the same for all sensors and, 
    \item a sensing matrix $D\in \{0,1\}^{n\times k}$ where $D_{i,j}=1$ means that sensor $s_j$ is capable of sensing target $t_i$.
\end{enumerate}
We assume that all of these parts are known to both players at any point in time. Note that in this paper, we consider the constant-sum utility structure and leave the general-sum version for future work.
That is, \texttt{Blue} 
seeks to maximize the summed value of \emph{not-sensed} targets, i.e., targets not sensed by any sensor. In contrast, \texttt{Red} seeks to minimize this value, or equivalently, maximize the summed value of targets that are sensed by some sensor.

The strategy of \texttt{Blue} is an ordering of the targets $\sigma: T\to [n]$ that assigns each target $t\in T$ a unique position $\sigma(t)$, i.e., $\sigma$ is a bijection. 
The targets move through the channel according to $\sigma$, i.e., the target on position $1$ moves first, on position $2$ second, and so on. In each \emph{time step} each target moves to the next sensor, leaves the channel (in case it passed all sensors), or enters the channel at the first sensor (in case it is the next target in the ordering $\sigma$).

The strategy of \texttt{Red} is a sensing plan $\psi: S \to 2^{T}$ that maps each sensor to a subset of targets $T'\subseteq T$ sensed by the sensor, where each sensor senses different targets, i.e., $\psi(s)\cap \psi(s')=\emptyset$ for each $s\neq s'\in S$. 
\texttt{Red} cannot play arbitrary sensing plans but only those which are \emph{valid}.
A sensing plan is \emph{valid} (with respect to a senor ordering $\sigma$) 
if 
\begin{enumerate*}[label=(\roman*)]
    \item a sensor only senses targets it has the capabilities to sense, i.e., for each $s_j\in S$ and $t_i\in \psi(s_j)$ we have $D_{i,j}=1$, and
    \item a sensor pauses for at least $\tau$ time steps after sensing a target, i.e., for each $s\in S$ and $t,t'\in \psi(s)$ we have $|\sigma(t)-\sigma(t')|>\tau$.
\end{enumerate*}
Given a sensing plan, we can immediately calculate the value of not-sensed targets as $v(\psi):=\sum_{t_i\in T\setminus \cup_{s\in S} \psi(s)} v_i$,  which quantifies \texttt{Blue}'s~utility.

\paragraph{Objectives and equilibrium} Due to the motivating applications of our interest,  this work adopts \texttt{Blue}'s perspective and analyzes sequential play in this game by assuming \texttt{Blue} moves first.\footnote{Note that this is already reflected in our game definition, since the validity of \texttt{Red}'s sensing plan depends on the strategy of \texttt{Blue}.  Thus, \texttt{Red} cannot move before or simultaneous to \texttt{Blue}.}
Our analysis consists of two parts. 
First, we will analyze the best response problem for \texttt{Red} called \textsc{Best Red Response}: Given an ordering $\sigma$ of the targets, output the sensing plan $\psi$ that is valid with respect to $\sigma$ and minimizes $v(\psi)$ among such plans.
Second, to compute the optimal strategy of \texttt{Blue}, we analyze the game's Stackelberg equilibrium\footnote{We assume that ties in the strategies are broken according to some predefined lexicographic ordering of the strategies.}, which can be written as the following bilevel optimization problem: $\max_{\sigma} \min_{\substack{\psi: \text{$\psi$ is valid wrt. $\sigma$}}} v(\psi).$
We term the corresponding computational problem \textsc{Blue Leader Stackelberg Equilibrium}.

\paragraph{A motivating application} One major motivation of our work is the secure transit of peacekeeping resources in the presence of adversarial actors such as pirates, which has critical importance due to past incidents, e.g., to the United Nations \cite{pirate-un}. Citing the UN's peacekeeping mission manual \cite{UNmarine}, ``\emph{protecting shipping in transit ensures the safety and security of vessels as they pass through waters threatened by piracy on the high seas...}'' In these applications, UN plays \texttt{Blue}'s role whereas pirates correspond to \texttt{Red}, who can observe the ordering of targets and then act second.  The UN commands a fleet of ships (i.e., targets in our model) that often carry resources of different importance and that can be arranged strategically. Protecting shipping is overall a complex, multi-facet, task and our model captures one of the phases after potential (often scarce) security measures have already been allocated to the ships and the pirates look to identify targets to attack. According to Winn and Govern  \cite{winn2008maritime},  pirates often use a set of boats (i.e., sensors in our model) to probe different passing targets, usually by following them to observe their speed, crew amount, firearm, etc. to judge based on this whether they are capable of capturing the ship. Such probing takes time, which is modeled by the recharging time $\tau$.  

\paragraph{Sensor and target types}
We develop some customized algorithms for instances with only a few different target or sensor models:
We say that two sensors $s_i,s_j\in S$ are of the same type if they are capable of sensing the same targets, i.e., the $i$-th and $j$-th column of the sensing matrix are identical. 
We say that two targets $t_i,t_j\in T$ are of the same type if they have the same utility value and can be sensed by the same sensors, i.e., $v_i=v_j$ and the $i$-th and $j$-th row of the sensing matrix are identical. 
We denote as $\Gamma=\{\gamma_1,\dots,\gamma_{n_{\chi}}\}$ and $\Theta=\{\theta_1,\dots,\theta_{k_{\chi}}\}$ the set of target and sensor types, respectively. 
It is easy to see that $k_\chi\leq 2^{n_\chi}$, as a sensor's sensing capabilities are defined by the set of target types it can sense. Similarly, assuming that all targets have the same value, it holds that $n_\chi \leq 2^{k_\chi}$.
\section{Related Work}
While the escape sensing game model is new, it is closely related to a few lines of AI research, as detailed below. 

\paragraph{Computational game theory for security} Conceptually, our work subscribes to the extensive MAS literature on computational game theory for tackling security challenges. The Stackelberg security game \cite{tambe2011security} is one widely studied example. Other game-theoretic models  include the hide-and-seek game \cite{chapman2014playing}, blotto games \cite{behnezhad2018battlefields}, auditting games \cite{blocki2015audit} and catcher-evader games \cite{li2016catcher}. Most of these games study the optimal usage of security forces under different game structures. In contrast, our ESG model is motivated by detection-vs-evasion situations in which security forces have already been allocated.

\paragraph{Scheduling} On a formal level,  our problem is to schedule/order targets in an adversarial environment, which shares similarities with the classic problem of \emph{scheduling} that looks to assign tasks to different machines to optimize certain criteria  \cite{lenstra1977complexity}. There is a rich body of AI research on scheduling, ranging from solving varied problems using AI techniques such as satisfiability \cite{crawford1994experimental} and distributed constraint optimization \cite{sultanik2007modeling}, to developing new models of scheduling problems under uncertainty \cite{bampis2022scheduling} or in multi-agent setups~\cite{zhang2016co}.

In fact, the \textsc{Best Red Response} problem can be formulated as the following slightly non-standard scheduling problem: There are $k$ machines (modeling sensors). In each step, a job (modeling a target) arrives. The job can be processed (modeling sensing) by a given subset of machines and if executed successfully generates a given reward value. The job has a processing time of $\tau$ and needs to be processed within the next $\tau$ steps. This implies that the job needs to be processed (i.e., sensed) either now or its reward is lost.

\section{The Algorithmics  of Escape Sensing Games}  \label{sec:ESG}
We analyze the computational complexity of ESGs starting with \texttt{Red}'s best response problem, followed by computing equilibria.

\subsection{Computing \texttt{Red}'s Best Response Strategy} \label{sec:compRed}
We analyze \texttt{Red}'s best response problem that \texttt{Red} needs to solve in each game after \texttt{Blue} has committed to a target ordering. 
This problem turns out to be NP-hard, even if \texttt{Red} is only interested in determining whether it can sense all targets. This intractability result is the first strong indicator of the intricate game dynamics in ESGs. 
\begin{restatable}{theorem}{redresponse}\label{redresponse-hardness}
\textsc{Best Red Response} is NP-complete, even when asked to decide whether \texttt{Red} can sense all targets or not.
\end{restatable}
\begin{proof}
We reduce from \textsc{Hitting Set} where we are given a universe $U$, a collection of sets $\mathcal{Z}=\{Z_1,\dots, Z_m\}$ and an integer $t$, and the question is whether there a size-$t$ subset $U'\subseteq U$ containing at least one element from each set in $ \mathcal{Z}$ (we assume that $t\geq 2$ and $|U|>t$). 
 
In the construction, all targets have a value of $1$ and the question is whether \texttt{Red} can sense all targets.
As the core of the construction we add \emph{element sensors} $\{a_u\mid u\in U\}$, \emph{set targets} $\{\alpha_{Z} \mid Z\in \mathcal{Z}\}$, and \emph{selection targets} $\{\beta_{i,j}\mid i\in [|U|-t], j\in [m]\}$. 
Each element sensor can sense all selection targets and all set targets corresponding to sets in which the element appears. 
Regarding the ordering of targets, it is easiest to think of the targets as being arranged in ``rounds''. In each round $j\in [m]$, first the selection targets $\{\beta_{i,j}\mid i\in [|U|-t]\}$ move through the channel followed by the the set target $\alpha_{Z_j}$. 
The idea is that the same $|U|-t$ element sensors sense the selection targets in every round, which correspond to the elements that are not part of the hitting set (we extend the construction in the following paragraph to ensure that this holds). 
Then, the remaining $t$ element sensors need to form a hitting set to be able to sense the set target in each round.

We extend the construction as follows. 
We add \emph{filling targets} $\gamma_{i,j}$ for all $i\in [t-1]$ and $j\in [m]$, which all element sensors can sense. 
Moreover, we add \emph{dummy sensors} $d_{i,j}$ for each $i\in [2|U|]$ and $j\in [m]$ and \emph{dummy targets} $\delta_{i,j}$ for each $i\in [2|U|]$ and $j\in [m]$. 
For each $i\in [2|U|]$ and $j\in [m]$, dummy sensor $d_{i,j}$ can sense dummy target $\delta_{i,j}$. 
We set $\tau:=2|U|+1$. 
Formally, the target ordering $\sigma$ is constructed---in multiple ``rounds''---as follows. In each round $j\in [m]$, we first move the selection targets $\{\beta_{i,j}\mid i\in [|U|-t]\}$ through the channel, then the dummy targets $\{\delta_{i,j}\mid i\in [|U|]\}$, then the set target $\alpha_{Z_j}$, then the filling targets $\{\gamma_{i,j}\mid i\in [t-1]\}$ and then the dummy targets $\{\delta_{i,j}\mid i\in [|U|+1,2|U|]\}$ (the ordering of targets in each of the groups is arbitrary). 

\paragraph{Proof of correctness: forward direction}
Assume that $U'\subseteq U$ is a size-$t$ hitting set of $\mathcal{Z}$. 
For each $i\in [2|U|]$ and $j\in [m]$, we let $d_{i,j}$ sense $\delta_{i,j}$. 
We construct the sensing plan for the element sensors iteratively as follows. 
In each round $j\in [m]$, we let each of the $|U|-t$ element sensors $\{a_u\mid u\in U\setminus U'\}$ sense exactly one of the selection targets $\{\beta_{i,j}\mid i\in [|U|-t]\}$. 
Now, let $u^*$ be an element from $U'$ that is contained in $Z_j$ (such an element needs to exist because $U'$ is a hitting set). 
We let $a_{u^*}$ sense $\alpha_{Z_j}$ and we let each of the $t-1$ element sensors  $\{a_u\mid u\in U' \setminus \{u^*\}\}$ sense exactly one of the filling targets $\{\gamma_{i,j}\mid i\in [t-1]\}$. 

The constructed sensing plan senses all targets and clearly respects the sensing matrix. 
It remains to be argued that the recharging times of all element sensors are respected (dummy sensors only sense one target).
For each $u\in U \setminus U'$, we have that $a_u$ senses one selection target in each round. Between two selection targets in two different rounds there are at least $2|U|$ dummy targets and one set target, so recharging times are respected. 
For each $u\in U'$, the sensor $a_u$ senses either a set or filling target in each round. There are $2|U|$ dummy targets and $|U|-t\geq 1$ selection targets between each two sets and filling targets from different rounds, so recharging times are respected. 

\paragraph{Proof of correctness: backward direction}
Assume that $\psi$ is a valid sensing plan that senses all targets. 
Consequently, in each round, the $|U|$ element sensors need to sense $|U|-t$ selection, $t-1$ filling, and one set target. 
As $\psi$ is valid and there are only $|U|-t-1+|U|+1+t-2=2|U|-2$ targets between the first selection and last filling target in each round, this means that each element sensor needs to sense exactly one of these targets in each round.
Note that an element sensor that senses a non-selection (i.e., either a set or filling) target in round $j\in [m]$ cannot sense a selection target in round $j+1$, as there are only $t-1+|U|+|U|-t-1=2|U|-2$ targets between the first non-selection target in round $j$ and the last selection target in round $j+1$. 
Consequently, as each element sensor needs to sense one target in each round, it follows   
that there is a set $U''\subseteq U$ of $|U|-t$ elements so that the corresponding element sensors sense a selection target in every round. 
Consequently, the remaining $t$ element sensors need to sense all set targets. 
As an element sensor is only capable of sensing a set target if the element appears in the set, it follows that $U\setminus U''$ is a size-$t$ hitting set of $\mathcal{Z}$.
\end{proof}

Despite this intractability result, it is still possible to construct exact combinatorial algorithms for  \textsc{Best Red Response}. 
In particular, we present a dynamic programming-based algorithm empowered by some structural observations on ESGs that runs in 
$\mathcal{O}( n\cdot (k_\chi +1)^{\tau+2})$ (recall that $k_\chi\leq k$). 
This algorithm in particular implies that the problem becomes polynomial-time solvable if the recharging time, which we expect to be rather small in comparison to the number of targets, is a constant. 
\begin{restatable}{proposition}{redresponseDP}\label{pr:redresponseDP}
There is a   $\mathcal{O}( n\cdot (k_\chi +1)^{\tau+2})$-time algorithm for \textsc{Best Red Response}.  
\end{restatable}
\begin{proof}[Proof Sketch]
Our idea is to construct a valid sensing plan iteratively by going through the arriving targets one by one (we assume that the ordering of targets is $t_1,\dots, t_n$). For each target, we either decide that it will not be sensed or assign it to one of the sensors so that the resulting plan is still valid.  
Our key observation to bring down the time and space complexity of the dynamic program is that we do not need to store the full sensing plan to ensure the validity of the plan after updating it. Instead, it is sufficient to know for each sensor whether it has sensed a target in the last $\tau$ steps. 
More formally, given a valid sensing plan $\psi$ that has been constructed by iterating over the first $i\in [n]$ targets, we only store the following information: 
\begin{enumerate*}[label=(\roman*)]
\item the value of all targets $\{t_1,\dots, t_i\}$ that have not been assigned to a sensor in $\psi$, i.e., $\sum_{t_j\in \{t_1,\dots, t_i\}\setminus \cup_{s\in S} \psi(s)} v_j$,
\item the sensors the last $\tau+1$ targets have been assigned to.
\end{enumerate*}
It is possible to store this information in a table of size $\mathcal{O}(n\cdot (k +1)^{\tau+1})$ where each cell can be computed in $\mathcal{O}(k)$-time.  To extend the algorithm to sensor types, we prove that we can collapse sensors of one type into a ``meta'' sensor, making it sufficient to bookmark the types of sensors that have sensed the last $\tau+1$ targets.
\end{proof}
We conclude by giving a clean ILP formulation of \textsc{Best Red Response}, which turns out to scale very favorably in our experiments allowing us to solve instances with up to $10000$ targets within one minute.
\begin{proposition}\label{pr:ILP-best}
    \textsc{Best Red Response}  admits  an ILP formulation with  $\mathcal{O}(n\cdot k)$ binary variables and $\mathcal{O}(n\cdot k)$ constraints. 
\end{proposition}
\begin{proof} 
    We model an instance $\mathcal{I}$ of \textsc{Best Red Response} as an ILP as follows. 
    We assume that the targets are ordered as $t_1,\dots, t_n$. 
    We create a binary variable $x_{i,j}$ for each $i\in [n]$ and $j\in [k+1]$. 
    Setting $x_{i,j}$ to one corresponds to letting sensor $s_j$ sense target $t_i$ if $j\in [k]$, and letting  $t_i$ not be sensed by any sensor if $j=k+1$. 

    To ensure that \texttt{Red} minimizes the value of not-sensed targets, the optimization criterion becomes: 
    $\min \sum_{i\in [n]} v_i\cdot x_{i,k+1}.$
    To ensure the validity of the sensing plan $\psi$, for each $i\in [n]$, we enforce that: 
    $\sum_{j\in [k+1]} x_{i,j}=1.$
    Moreover, to ensure that sensor capabilities are respected, we impose for each $i\in [n]$ and $j\in [k]$ that: 
    $x_{i,j}\leq D_{i,j}.$
    Lastly, to enforce that recharging times are respected, for each $j\in [k]$ and $i\in [n-\tau]$ we add the constraint:
    $\sum_{\ell=i}^{i+\tau} x_{\ell,j}\leq 1.$
\end{proof}

\subsection{Solving for the Stackelberg Equilibrium}\label{sec:Stackle}
We now study the problem of computing \texttt{Blue}'s optimal strategy, i.e., to solve \textsc{Blue Leader Stackelberg Equilibrium}. \Cref{redresponse-hardness} already shows the NP-hardness of \textsc{Best Red Response}. While this does not imply the hardness of computing Stackelberg equilibria\footnote{Note that the fact that it is NP-hard for \texttt{Red} to best respond to certain \texttt{Blue} strategies (as constructed in the reduction of \Cref{redresponse-hardness}) does not imply that is also hard for \texttt{Red} to best respond to the particular Stackelberg equilibrium strategy of \texttt{Blue} (as these strategies might admit some structure that makes it easier to best respond).}, a convincing intractability result for \texttt{Blue}'s optimal strategy shall ideally ``disentangle'' its complexity from \texttt{Red}'s best response problem. With this in mind, we prove the NP-hardness of \textsc{Blue Leader Stackelberg Equilibrium}  even in situations where  \texttt{Red}'s best response problem is linear-time solvable. 
This demonstrates that the complexity in our reduction does not come from finding \texttt{Red}'s strategy but from the problem of whether \texttt{Blue} can arrange the targets in an optimal way.  

\begin{restatable}{theorem}{stackcomp}
\textsc{Blue Leader Stackelberg Equilibrium} is NP-hard, even on instances where  \textsc{Best Red Response} is linear-time computable and the recharging time is $3$. 
\end{restatable}
Note that the NP-hardness upholds even if sensors' recharging time is constant, a case in which \texttt{Red}'s best response problem is polynomial-time solvable (see \Cref{pr:redresponseDP}).
Our hardness result indicates that computing \texttt{Blue}'s optimal strategy is a generally much harder problem than computing  \texttt{Red}'s optimal strategy.
In fact, it remains open whether \textsc{Blue Leader Stackelberg Equilibrium} is contained in NP or whether it is complete for complexity classes beyond NP. We suspect the latter to hold.

\subsubsection{Bilevel Optimization}
In light of this, it is unclear (and from our perspective rather unlikely) that \textsc{Blue Leader Stackelberg Equilibrium} admits an ILP formulation. 
Naive brute-force approaches are also computationally infeasible, as we would need to enumerate all $n!$ possible target orderings and solve the NP-hard \textsc{Best Red Response} problem as a subroutine for each of them. 

Thus, we turn to a formulation as a bilevel optimization problem \cite{colson2007overview} as one way to solve the problem exactly. 
In such formulations, constraints are still linear,
but there exist two connected levels of the problem, i.e., an outer and an inner level. The inner level controls certain variables that it sets to minimize an objective subject to linear constraints that also involve variables controlled by the outer level, while the outer level sets these variables to maximize the objective. 
In our problem, we can model \texttt{Red}'s best response problem as the inner level loosely following the ILP from \Cref{pr:ILP-best}. 
The outer-level models \texttt{Blue}'s problem. 
The key parts of the outer level are variables for each target that encode the position in which the target appears in the final ordering and that are used in the inner level to ensure the validity of the sensing plan. 
\begin{restatable}{proposition}{bilevel}\label{pr:bilevel}
\textsc{Blue Leader Stackelberg Equilibrium} admits a bilevel optimization formulation with $\mathcal{O}(n^2+n\cdot k)$ binary variables, $\mathcal{O}(n)$ integer variables, and $\mathcal{O}(n^2\cdot k)$ constraints.  
\end{restatable}
Note that standard techniques to convert this bilevel program into an (integer) linear program, e.g., by exploiting KTT-optimality conditions \cite{allende2013solving,dempe2020bilevel}, are not applicable in our setting, as we are solving an \emph{integer} bilevel program within which the inner-level program is already non-convex.

\subsubsection{Heuristic}
We will see later that the running time for the bilevel formulation of the problem becomes already infeasible on small-sized instances.
Therefore, we experimented with different heuristics to solve the problem.\footnote{Note that the heuristic double-oracle approach that has been successfully employed for other large combinatorial games \cite{DBLP:conf/aaai/AdamHKK21,DBLP:conf/atal/JainKVCPT11} is not applicable to ESGs. Traditionally, the approach successively expands the strategy spaces of both players by letting them best respond to each other. However, in ESGs,  we face a bilevel problem in which there is no best response of the leader to the follower. The approach also fails here because the valid strategies of the follower heavily depend on the strategy picked by the leader.}
In the following, we present two variants of 
simulated annealing-based heuristics that performed best. 
For a target ordering $\sigma$, we denote as $N(\sigma)$ its neighbors, i.e., all  ${n \choose 2}$ orderings that arise from $\sigma$ by swapping the position of any two different targets. 
The \emph{relaxed version} of our simulated annealing (SA\_Relax) is presented in \Cref{alg:SA_Relax}. 
The idea is to find an optimal ordering through repeated local rearrangements.
We store the current ordering as $\sigma$ and compute its value for \texttt{Blue} by solving \texttt{Red}'s best response problem using \Cref{pr:ILP-best}. 
Then, we pick a random neighbor of $\sigma$, compute its value, and update the ordering based on this according to standard simulated annealing~rules.

\begin{algorithm}[t]
\begin{small}
\caption{\textsc{SA\_Relax}}
\label{alg:SA_Relax}
\textbf{Input}: Target ordering $\sigma$ and temperature~$T=100$

\begin{algorithmic}[1] 
\STATE Compute optimal sensing plan $\psi$ wrt. $\sigma$. 
\WHILE{$T > 0.00001$}
\STATE Select a random neighbor $\hat{\sigma} \in N(\sigma)$.\label{line:rand}
\STATE  Compute optimal sensing plan $\psi'$ wrt. $\hat{\sigma}$.
\IF{$e^{\frac{v(\hat{\psi})-v(\psi)}{T}} >$ random[0, 1]}
\STATE $\sigma := \hat{\sigma}$, $\psi:=\hat{\psi}$.
\ENDIF
\STATE $T: = T\cdot 0.9$.
\ENDWHILE
\STATE \textbf{Return} $\sigma$. 
\end{algorithmic}
\end{small}
\end{algorithm}

In the \emph{full version} of our simulated annealing (SA), instead of picking a random neighbor $\hat{\sigma}$ from $N(\sigma)$ in Line \ref{line:rand} of Algorithm \ref{alg:SA_Relax}, we first run a heuristic for \textsc{Best Red Response}
on all orderings from $N(\sigma)$.\footnote{In our (greedy) heuristic, we consider the targets in decreasing order of their value and construct the sensing plan $\psi$ iteratively. Let $S'\subseteq S$ be the sensors $s$ so that $\psi$ remains a valid sensing plan after adding the current target to $\psi(s)$. We let the target be sensed by a randomly selected sensor from $S$ (or by no sensor if $S$ is empty). For a formal description, see our full version \cite{us}.}
Then, on the $\mu$ fraction of neighbors with the highest returned value, we execute the ILP from \Cref{pr:ILP-best} to compute the optimal sensing plan. Of the examined neighbors, we pick the one with the highest returned value as $\sigma'$. As a hyper-parameter tuning process, we tested the performance of our heuristic algorithm with respect to the choice of $\mu$ (see our full version \cite{us} for results). 
It turns out that $\mu = 0.1$ provides a good trade-off between the algorithm's running time and \texttt{Blue}'s utility. Thus, we fix $\mu=0.1$ throughout the paper.
For both heuristics, we always run the heuristic three times with three different initial randomly generated target orderings and return the best computed ordering.  

\section{Experimental Evaluations}\label{sec:exp_bilevel}

We analyze the quality and performance of our algorithms to compute the Stackelberg equilibrium.\footnote{We use Gurobi \cite{gurobi} to solve the ILP from \Cref{pr:ILP-best} and MIBS \cite{tahernejad2020branch} to solve the bilevel program from \Cref{pr:bilevel}. Both are among the most popular off-the-shelf tools for solving the respective problem.} 
We consider three simulated game settings for generating ESGs. For each setting, we
determine the value of a target by drawing a number uniformly within $[0, 1]$\footnote{In our full version \cite{us}, we analyze  supplementary scenarios, reinforcing similar conclusions to those presented~here.}: 

\begin{enumerate}
    \item {\bfseries\scshape Default (Def):} For each $i\in [n]$ and $j\in [k]$, we set $D_{i,j} =1$ with probability $0.2$.
    \item {\bfseries\scshape Euclidean (Euc):}  
    Each target $t_i \in T$ and each sensor $s_j \in S$ are uniformly sampled points in $[0,1] \times [0,1]$. A sensor can sense a target (i.e., $D_{i,j} = 1$) if the Euclidean distance between their points is below $0.3$.
    \item {\bfseries\scshape RandomLevel (Rand):}  
    Each target $t_i\in T$ has a difficulty level $d_i$ uniformly sampled from $[0,1]$, and each sensor $s_j\in S$ has a skill level $s_j$ uniformly sampled from $[0,1]$.
    For each $i\in [n]$ and $j\in [k]$, we set $D_{i,j} =1$ with probability $(1-d_i)\cdot s_j$.
\end{enumerate}
In all our experiments, if not stated otherwise, we average over $50$ instances generated according to one of the models. 
We present our experimental results as tables where each entry contains  \texttt{Blue}'s average utility (i.e., the summed value of not-sensed targets) from the computed target ordering assuming \texttt{Red} best responds and the average running time in seconds in \textit{italics}, both followed by their respective standard deviations.
Note that standard deviations are calculated across the different sampled instances, implying that independent of the solution method some non-trivial standard deviation is to be expected, as certain instances are more favorable for \texttt{Blue} than others. 

We analyze the maximum size of instances that we can solve exactly using the bilevel program, which we denote as OPT. 
We present results for the {\bfseries\scshape Default} game setting in \Cref{tab:scability_bilevel} (results for other simulated game settings are similar).
It turns out that while instances with $5$ targets can be solved within a second by OPT, instances with $9$ targets take already around $9$ hours to solve. This demonstrates that the bilevel program is only usable for quite small instances.
Moreover, we observe the to-be-expected trend that \texttt{Blue}'s utility increases when \texttt{Blue} has more targets or \texttt{Red} has less sensors. However, we do not find any consistent trend regarding whether it is more advantageous for \texttt{Blue}: more targets or fewer sensors. 

Motivated by the high computational cost of the bilevel program, we now turn to analyzing the quality of our heuristics. We also include the \emph{Random} method here as a baseline where \texttt{Blue} simply picks an arbitrary ordering of targets (and \texttt{Red} best responds to it). In addition, we compare our heuristics against a naive random strategy of comparable computational cost. For this, we include the \textit{Random2} method which generates $3000$ random orderings for Table \ref{tab:bilevel_approx1} and $3\cdot 10^7$ random orderings for Table \ref{tab:bilevel_approx2}. The sampled ordering that achieves the highest utility for \texttt{Blue} assuming that \texttt{Red} best responds is returned.

\begin{small}
	\begin{table*}[t]
		\centering
		\begin{minipage}[t]{0.32\linewidth}
			\vspace{0pt}
			\resizebox{\textwidth}{!}{ \begin{tabular}{|c|c|c|c|}
					\hline
					\diagbox{\small  \#targets}{\small  \#sensors} &  2 & 3 & 5\\
					\hline
					5 & \makecell{1.79 $\pm$ 0.71, \\ \textit{0.61 $\pm$ 0.21}}&  \makecell{ 1.55 $\pm$ 0.65, \\ \textit{0.77 $\pm$ 0.02}}& \makecell{ 1.02$\pm$ 0.62, \\ \textit{0.98 $\pm$ 0.04}}\\
					\hline
					7 & \makecell{2.41 $\pm$ 0.78, \\ \textit{102 $\pm$ 36}}&  \makecell{ 2.29 $\pm$ 0.80, \\ \textit{116 $\pm$ 31}}& \makecell{ 1.62$\pm$ 0.72, \\ \textit{140 $\pm$ 29}}\\
					\hline 
					8 & \makecell{2.96 $\pm$ 0.81, \\ \textit{1501 $\pm$ 354}}&  \makecell{ 2.33 $\pm$ 0.74, \\ \textit{1760 $\pm$ 38}}& \makecell{1.7 $\pm$ 0.69, \\ \textit{1814 $\pm$ 23} }\\
					\hline 
					9  & \makecell{ n/a, \\ \textit{31358}}&  \makecell{ n/a, \\ \textit{32541}}& \makecell{ n/a, \\ 35376}\\
					\hline
			\end{tabular}}
			\caption{Scalability test of bilevel-program (OPT) for  {\bfseries\scshape Default} game setting with $\tau=2$. 
				For $n=9$, we report running time for one instance. For all tables: each entry shows \texttt{Blue}'s average utility (top) and running time in seconds (bottom). 
			}
			\label{tab:scability_bilevel}
		\end{minipage}
		\hfill
		\begin{minipage}[t]{0.32\linewidth}
			\vspace{0pt}
			\resizebox{\textwidth}{!}{ \begin{tabular}{|c|c|c|c|}
					\hline
					\diagbox{\small  Algo.}{\small  Setting}   & {\bfseries\scshape Def}  & {\bfseries\scshape Euc} & {\bfseries\scshape Rand}  \\
					\hline 
					OPT & \makecell{\textbf{2.29} $\pm$ 0.80, \\\textit{ 116 $\pm$ 31}} & \makecell{\textbf{1.952} $\pm$ 0.74, \\ \textit{ 126$\pm$ 2.29}} &  \makecell{\textbf{2.09} $\pm$ 0.87, \\ \textit{120 $\pm$ 25}} \\
					\hline 
					SA & \makecell{\textbf{2.29} $\pm$ 0.80, \\ \textit{4.78 $\pm$ 0.38}} & \makecell{1.951 $\pm$ 0.75, \\ \textit{4.96 $\pm$ 0.47}} & \makecell{\textbf{2.09} $\pm$ 0.87, \\ \textit{5.03 $\pm$ 0.69}}\\
					\hline 
					SA\_Relax &  \makecell{\textbf{2.29} $\pm$ 0.80, \\ \textit{0.81 $\pm$ 0.03}} & \makecell{\textbf{1.952} $\pm$ 0.74, \\ \textit{0.84 $\pm$ 0.05}}& \makecell{\textbf{2.09} $\pm$ 0.87, \\ \textit{0.85 $\pm$ 0.08}}\\
					\hline
					Random & \makecell{2.16 $\pm$ 0.84, \\ \textit{0.001}} & \makecell{1.71 $\pm$ 0.83, \\ \textit{0.001}} & \makecell{1.93 $\pm$ 0.92, \\ \textit{0.001}}\\
					\hline
					Random2 & \makecell{\textbf{2.29} $\pm$ 0.80, \\ \textit{5.13 $\pm$ 0.16}} & \makecell{\textbf{1.952} $\pm$ 0.74, \\ \textit{5.25 $\pm$ 0.33 }} & \makecell{\textbf{2.09} $\pm$ 0.87, \\ \textit{5.3 $\pm$ 0.46}}\\
					\hline
			\end{tabular}}
			\caption{Comparison of algorithms to compute \texttt{Blue}'s utility for different simulated game settings, where $n=7$, $k=3$, and $\tau=2$.} 
			\label{tab:bilevel_approx1}
		\end{minipage}\hfill
		\begin{minipage}[t]{0.32\linewidth}
			\vspace{0pt}
			\resizebox{\textwidth}{!}{\begin{tabular}{|c|c|c|c|}
					\hline
					\diagbox{\small  Algo.}{\small  Setting}  & {\bfseries\scshape Def}  & {\bfseries\scshape Euc} & {\bfseries\scshape Rand} \\
					\hline 
					SA  & \makecell{\textbf{15.96} $\pm$ 1.1, \\ \textit{28101 $\pm$ 563}} & \makecell{\textbf{18.4} $\pm$ 2.1, \\ \textit{27755 $\pm$ 928}} & \makecell{\textbf{17.3} $\pm$ 4.2, \\ \textit{27970 $\pm$ 1136}}\\
					\hline
					SA\_Relax  & \makecell{8.76 $\pm$ 0.9, \\  \textit{49.6 $\pm$ 2.57}}& \makecell{10.53 $\pm$ 2.32, \\  \textit{47.5 $\pm$ 0.86}} & \makecell{12.3 $\pm$ 3.6, \\  \textit{49.7 $\pm$ 1.6}} \\
					\hline
					Random   & \makecell{6.19 $\pm$ 1.26,  \\ \textit{0.001}} & \makecell{6.86 $\pm$ 2.25,  \\ \textit{0.001}} & \makecell{9.68 $\pm$ 2.78,  \\ \textit{0.001}}\\
					\hline
					Random2   & \makecell{8.27 $\pm$ 0.73,  \\ \textit{25036 $\pm$ 311}} & \makecell{9.54 $\pm$ 2.45,  \\ \textit{24333 $\pm$ 211}} & \makecell{11.68 $\pm$ 3.58,  \\ \textit{26810 $\pm$ 295}}\\
					\hline
			\end{tabular}}
			\caption{Comparison of algorithms to compute \texttt{Blue}'s utility for different simulated game settings, where $n=75$, $k=10$, and $\tau=5$. We generate $9$ instances per method.} 
			\label{tab:bilevel_approx2}
		\end{minipage}
	\end{table*}
\end{small}

In \Cref{tab:bilevel_approx1},  we show the algorithms' performance for small instances where we can still compute \texttt{Blue}'s maximum utility (OPT) via the bilevel program. 
In \Cref{tab:bilevel_approx2}, we consider larger instances where the optimum value is unknown. 
Note that higher values correspond to a better performance of the algorithm, as we always report \texttt{Blue}'s  utility for \texttt{Red}'s best response.

From the results in \Cref{tab:bilevel_approx1}, we can see that all heuristics perform well on small instances. In particular, SA\_Relax, SA, and Random2 find the optimal solution in all (but one) cases. However, SA\_Relax proves advantageous because it only needs a sixth of the running time of the other two methods.

While our two heuristics SA and SA\_Relax show a similar approximation quality for small instances, for larger instances (\Cref{tab:bilevel_approx2}) SA clearly outperforms SA\_Relax. For the {\bfseries\scshape Default} game setting, using SA compared to SA\_Relax even regularly leads to a doubled utility for \texttt{Blue}. 
While this is a strong argument for using SA, SA's downside is its higher computational cost, needing over $7$ hours to solve instances with $75$ targets. 

Finally, we observe that both methods clearly outperform the Random baseline, with SA consistently preserving an average of approximately 20 more targets for the larger instances.
This highlights that the solution quality of the target ordering clearly increases throughout the simulated annealing. 
Considering Random2, we find that repeatedly sampling orders (instead of only once) leads to a noticeable utility increase. However, on the larger instances, Random2 performs even worse than SA\_Relax while running as long as SA, thereby combining the disadvantages of SA and SA\_Relax.
Overall, our experiments highlight that \texttt{Blue} benefits from ordering the targets strategically instead of randomly.

\section{Escape from Non-Coordinated Sensing}\label{non-coor}

ESGs assume that the different sensors are controlled by a central authority that computes the sensing plan. We now investigate the situation where these sensors are \emph{non-coordinated} and each one acts independently based on a natural greedy algorithm. This happens when sensors cannot easily exchange information and coordinate with each other. Another motivation is when  
 sensors are controlled by different adversaries, each serving only their own interests and being unlikely to coordinate their actions and share their reward.  
 Both of these scenarios can occur in our motivating domain of piracy at large open seas, as coordination between different groups is likely to be challenging. Different pirate groups might even refuse to coordinate at all and instead directly compete with each other.
 
We model these situations by assuming that sensors have a predefined ordering as $s_1,\dots,s_k$ (as induced by fixed locations of the sensors); for each sensor $s\in S$, as soon as a not previously sensed target that $s$ can sense passes $s$ (i.e., $s$ has the capabilities and is currently not recharging), $s$ senses it, thereby greedily maximizing its number of sensed targets.

Formally, given a target ordering $\sigma$, we construct a sensing plan $\psi_\sigma$ sequentially as follows. For each step $\ell\in [n+k-1]$, if target $t_i$ passes sensor $s_j$ in step $\ell$, then we add $t_i$ to $\psi_\sigma(s_j)$ if the resulting sensing plan remains valid with respect to $\sigma$ (formally, for $i\in [n]$ and $j\in [k]$ target $\sigma^{-1}(i)$ passes sensor $s_j$ in step $i+j-1$). 
As the strategy of  \texttt{Red} is fixed, the problem \textsc{Best Blue Response} \texttt{Blue} faces is to pick a target ordering $\sigma$ so that $v(\psi_\sigma)$ gets maximized.
In the following, we study the computational complexity of this problem and solve it in computational experiments. 
By comparing the answer of \textsc{Best Blue Response} to the value of the Stackelberg equilibrium in the corresponding ESG we can ultimately answer how much \texttt{Red} gains from being able to centrally control its sensors. 

\subsection{Algorithmic Analysis}\label{sec:NonCoordAA}

Unfortunately, it turns out that computing \texttt{Blue}'s strategy is NP-hard, even in restricted cases where each sensor can only sense one target. 
Due to the sequential construction of \texttt{Red}'s sensing plan, this reduction is our most intricate one:

\begin{restatable}{theorem}{bestblue}\label{th:bestblue}
\textsc{Best Blue Response} is NP-complete, even if the recharging time is $\infty$, i.e., each sensor can sense only one target, each target has value $1$, and the sum of each row and column in the sensing matrix is at most four. 
\end{restatable}

\begin{proof}[Proof Sketch]
We focus on the variant where each sensor can only sense one target. 
Interestingly, as discussed in more detail in our full version \cite{us} this problem shares some similarities with the NP-hard \textsc{Minimum Maximal Matching} problem, as we can view the sensors and targets as two sides of a bipartite graph with sensor-target pairs where the sensor senses the target corresponding to maximal matchings in this graph. 
However, the ordering of the sensors makes only certain maximal matchings in these graphs realizable, which is why we instead show NP-hardness by reducing from a variant of 3-SAT where each variable appears only twice positive and once negative. 
The core idea of our construction is the following: We add a \emph{literal target} for each literal. Moreover, for each clause, we add a \emph{clause sensor} and a \emph{clause target}. The clause sensor is capable of sensing the corresponding clause target as well as targets corresponding to the three literals appearing in the clause. 
We add further targets and sensors to the instance so that all clause targets need to make it unsensed through the channel.  
This implies that each clause sensor needs to sense a literal target as it will otherwise sense the corresponding clause target in passing, i.e., we need to ``cover'' each clause with a literal appearing in the clause. 
Now for each variable, we add a slightly intricate gadget that ensures that we can either use the targets corresponding to positive literals to cover clause sensors (which corresponds to setting the variable to true) or the one target corresponding to a negative literal (which corresponds to setting the variable to false). Because we need to ``cover'' each clause, the induced assignment is satisfying. 
\end{proof}

We can adopt a similar view as in \Cref{pr:redresponseDP} to solve the problem via dynamic programming. 
However, this time the dynamic programming iteratively constructs the optimal target ordering and we need to keep track of the previously used targets together with the sensors used in the last $\tau+1$ timesteps. This results in a naive running time of $\mathcal{O}(n\cdot 2^n\cdot (k+1)^{\tau+2})$, which can be improved to  $\mathcal{O}\left(n_\chi\cdot \left(\prod_{i=1}^{n_\chi} (\ell_i+1)\right)\cdot (k+1)^{\tau+2}\right)$ if we incorporate types:
\begin{restatable}{proposition}{DBGreedy}\label{bbr:DP}
\textsc{Best Blue Response} is solvable in $\mathcal{O}\left(n_\chi\cdot\left(\prod_{i=1}^{n_\chi} (\ell_i+1)\right)\cdot (k+1)^{\tau+2}\right)$, where $\ell_i$ is the number of targets of type $\gamma_i$.
\end{restatable}

\subsubsection{ILP Formulation}\label{sec:BestILP}
Constructing an ILP for \textsc{Best Blue Response} turns out to be slightly more challenging, as we need to encode \texttt{Red}'s greedy sequential behavior:
\begin{restatable}{proposition}{bestblueILP}\label{pr:bestblueILP}
    There is an ILP formulation for \textsc{Best Blue Response} with  $\mathcal{O}(n^2\cdot k)$ binary variables, $\mathcal{O}(n)$ integer variables, and $\mathcal{O}(n^2\cdot k)$ constraints.
\end{restatable}
\begin{proof}[Proof Sketch]
    We introduce for each target $i\in [n]$ an integer variable $z_i$ encoding the position in which the target appears. 
    Moreover, similar to  \Cref{pr:ILP-best}, for each $i\in [n]$ and $j\in [k+1]$, we add a binary variable $x_{i,j}$, which encodes whether $t_i$ is sensed by sensor $s_j$ or whether the target makes it unsensed through the channel (for $j=k+1$).
    We can add mostly straightforward constraints to ensure that $x_{i,j}$ respects recharging times. 
    The main challenge is to encode the greedy behavior of the sensors (i.e., the ILP cannot have the freedom to pick the $x_{i,j}$ values arbitrarily to optimize \texttt{Blue}'s utility but they are set according to sensors' greedy behavior). 
    For this, for each $i,i'\in [n]$ and $j\in [k]$, we add a binary variable $y_{i,i',j}$ and add constraints so that $y_{i,i',j}$  is equal to one if target $i$ is sensed by sensor $j$ and because of this $j$ recharges when $i'$ is passing, i.e., $i$ ``covers'' $i'$. 

To encode sensors' greedy behavior, we want to add a constraint that makes sure that in case $x_{i,j}=1$, the target needs to be covered by other targets for all sensors that are capable of sensing it placed before $j$. Note that this together with another constraint ($\sum_{j\in [k+1]} x_{i,j}=1$) in particular implies that each target is sensed by the first sensor it passes which is not recharging, thereby encoding the greedy behavior of sensors. 
Specifically, for each $i\in [n]$ and $j\in [k+1]$, we add: 
\begin{small}
\begin{align} \label{eq:prob}
    \sum_{t\in [j-1]: D_{i,t}=0} 1+ &\sum_{t\in [j-1]:D_{i,t}=1} \sum_{i'\in [n]} y_{i',i,t} -(j-1)\\ \nonumber &  \geq -n(1-\sum_{t=j}^{k+1} x_{i,t}).
\end{align}
\vspace{-0.6cm}
\end{small}
\end{proof}

\subsubsection{Heuristic}

Since it will turn out that the ILP formulation cannot quickly solve medium-to-large instances, we explore various simulated annealing-based heuristics, similar to the approach discussed in Section \ref{sec:exp_bilevel}. 
We present the variant SA\_Relax where a random neighbor is picked in \Cref{alg:SA_Blue_Best_Response}.
The other variant SA computes \texttt{Blue}'s utility $v(\psi_{\hat{\sigma}})$ for all neighbors and picks the one with the highest utility.  
\begin{algorithm}[t]
\caption{SA\_Relax for \textsc{Best Blue Response}}
\label{alg:SA_Blue_Best_Response}
\begin{small}
\textbf{Input}: Initial target ordering $\sigma$ and temperature $T=100$

\begin{algorithmic}[1] 
\WHILE{$T > 0.00001$}
\STATE Select a random neighbor $\hat{\sigma} \in N(\sigma)$.
\IF{$e^{\frac{v(\psi_{\hat{\sigma}})-v(\psi_{\sigma})}{T}} >$ random[0, 1]}
\STATE $\sigma := \hat{\sigma}$.
\ENDIF
\STATE $T := T\cdot 0.9$.
\ENDWHILE
\STATE \textbf{Return} $\sigma$. 
\end{algorithmic}
\end{small}
\end{algorithm}
 
\subsection{Experiments}

We reuse the general setup described in \Cref{sec:exp_bilevel}, but naturally now report \texttt{Blue}'s computed utility assuming that sensors act greedily. Here, we let the Random2 method generate $1000$ random orderings in Table \ref{tab:blue_best_response_approx1} and $5\cdot 10^5$ random orderings in Table \ref{tab:blue_best_response_approx2}.

First of all, we evaluate the scalability of our ILP for \textsc{Best Blue Response} (OPT) in Table \ref{tab:scability_blue_team_response}. The ILP can solve the problem for medium-sized instances with up to $25$ targets in a few minutes. However, due to the complexity of the ILP modeling, already for $25$ targets as soon as the number of sensors reaches $5$, instances can take more than $5$ hours to solve. This is why the last line of the table only reports the running time for one instance.

\begin{small}
	\begin{table*}[t]
		\centering
		\begin{minipage}[t]{0.32\linewidth}
			\vspace{0pt}
			\resizebox{\textwidth}{!}{ \begin{tabular}{|c|c|c|c|}
					\hline
					\diagbox{\small  \#targets}{\small  \#sensors} &  2 & 3 & 5\\
					\hline
					5 & \makecell{2 $\pm$ 0.68, \\ \textit{0.007 $\pm$ 0.005}}&  \makecell{ 1.71 $\pm$ 0.57, \\ \textit{0.009 $\pm$ 0.005}}& \makecell{ 1.32$\pm$ 0.52, \\ \textit{0.01 $\pm$ 0.003}}\\
					\hline 
					15  & \makecell{ 6.6 $\pm$ 1, \\ \textit{0.14 $\pm$ 0.35}}&  \makecell{ 6.29 $\pm$ 0.9, \\ \textit{0.32 $\pm$ 0.88}}& \makecell{ 5.46 $\pm$ 0.87, \\ \textit{6.85} $\pm$ 22.9}\\
					\hline
					20  & \makecell{ 9.05 $\pm$ 1.13, \\ \textit{0.52 $\pm$ 2.94}}&  \makecell{ 8.49 $\pm$ 1.01, \\ \textit{1.58 $\pm$ 4.4}}& \makecell{ 7.56 $\pm$ 1.01, \\ \textit{229 $\pm$ 1261}}\\
					\hline
					25  & \makecell{ 11.38 $\pm$ 1.26, \\ \textit{6.8 $\pm$ 30.6}}&  \makecell{ 10.96 $\pm$ 1.29, \\ \textit{283 $\pm$ 1771}}& \makecell{n/a, \\ \textit{22537}}\\
					\hline
			\end{tabular}}
			\caption{Scalability test of ILP (OPT) for {\bfseries\scshape Default} game setting with $\tau=2$. For all tables: each entry shows \texttt{Blue}'s average utility (top) and running time in seconds (bottom).}
			\label{tab:scability_blue_team_response}
		\end{minipage}
		\hfill
		\begin{minipage}[t]{0.32\linewidth}
			\vspace{0pt}
			\resizebox{\textwidth}{!}{
				\begin{tabular}{|c|c|c|c|}
					\hline
					\diagbox{\small  Algo.}{\small  Setting}  & {\bfseries\scshape Def}  & {\bfseries\scshape Euc} & {\bfseries\scshape Rand}  \\
					\hline
					OPT & \makecell{ \textbf{3.1}$\pm$ 0.86, \\ \textit{0.15 $\pm$ 0.37}} & \makecell{\textbf{3.25}$\pm$ 0.79, \\ \textit{0.25 $\pm$ 0.71}} &  \makecell{\textbf{3.04} $\pm$ 0.78, \\ \textit{3.26 $\pm$ 8.77}} \\
					\hline
					SA & \makecell{2.83 $\pm$ 0.81, \\ \textit{0.7$\pm$ 0.02}} & \makecell{2.92 $\pm$ 0.81, \\ \textit{0.71 $\pm$ 0.02}} & \makecell{2.72 $\pm$ 0.8, \\ \textit{0.71 $\pm$ 0.018}}\\
					\hline 
					SA\_Relax &  \makecell{3.06 $\pm$ 0.88, \\ \textit{0.02}} & \makecell{3.17 $\pm$ 0.81, \\ \textit{0.02}}& \makecell{2.89 $\pm$ 0.82, \\ \textit{0.02}}\\
					\hline
					Random & \makecell{1.9 $\pm$ 0.72, \\ \textit{0.001}} & \makecell{2.15 $\pm$ 0.9, \\ \textit{0.001}} & \makecell{1.9 $\pm$ 0.9, \\ \textit{0.001}}\\
					\hline
					Random2 & \makecell{2.91 $\pm$ 0.87, \\ \textit{0.81 $\pm$ 0.0004}} & \makecell{3.11 $\pm$ 0.79, \\ \textit{0.81 $\pm$ 0.0004}} & \makecell{2.89 $\pm$ 0.78, \\ \textit{0.83 $\pm$ 0.0004}}\\
					\hline
			\end{tabular}}
			\caption{Comparison of algorithms for \textsc{Best Blue Response} for different game settings, where $n=10$, $k=5$, and $\tau=2$.
			}
			\label{tab:blue_best_response_approx1}
		\end{minipage}
		\hfill
		\begin{minipage}[t]{0.32\linewidth}
			\vspace{0pt}
			\resizebox{\textwidth}{!}{\begin{tabular}{|c|c|c|c|}
					\hline
					\diagbox{\small  Algo.}{\small  Setting}   & {\bfseries\scshape Def}  & {\bfseries\scshape Euc} & {\bfseries\scshape Rand}  \\
					\hline 
					SA &  \makecell{\textbf{16.57} $\pm$ 1.64, \\ \textit{485 $\pm$ 20}} & \makecell{\textbf{17.2} $\pm$ 2.4, \\ \textit{470 $\pm$ 9.6}}& \makecell{\textbf{17.54} $\pm$ 3.27, \\ \textit{503 $\pm$ 20}}\\
					\hline
					SA\_Relax & \makecell{12.3 $\pm$ 1.58, \\ \textit{3.14 $\pm$ 0.35}} & \makecell{13.11 $\pm$ 2.62, \\ \textit{2.84 $\pm$ 0.2}} & \makecell{14.47 $\pm$ 2.96, \\ \textit{2.95 $\pm$ 0.2}}\\
					\hline 
					Random & \makecell{9.2 $\pm$ 1.99, \\ \textit{0.001}} & \makecell{10.02 $\pm$ 2.77, \\ \textit{0.001}} & \makecell{11.67 $\pm$ 3.07, \\ \textit{0.001}}\\
					\hline
					Random2 & \makecell{12.75 $\pm$ 1.03, \\ \textit{458 $\pm$ 15}} & \makecell{12.98 $\pm$ 2.37, \\ \textit{437 $\pm$ 13 }} & \makecell{14.67 $\pm$ 3.18, \\ \textit{496 $\pm$ 19}}\\
					\hline
			\end{tabular}}
			\caption{Comparison of algorithms for \textsc{Best Blue Response} for different game settings, where $n=75$, $k=10$, and $\tau=5$. 
			}
			\label{tab:blue_best_response_approx2}
		\end{minipage}
	\end{table*}
\end{small}

Next, we analyze the solution quality of our heuristic approaches. On small instances presented in \Cref{tab:blue_best_response_approx1}, our best heuristic algorithm approximates the optimal solution quite well and the error is typically below $10\%$ with the SA\_Relax method consistently outperforming SA. 
Both heuristics outperform Random, while Random2 performs better than SA (yet still worse than SA\_Relax, while having a much longer running time).
When moving to larger instances in \Cref{tab:blue_best_response_approx2}, the picture flips, as SA is now substantially outperforming SA\_Relax. This shows a general trend that the solution quality of SA scales more favorably than that of SA\_Relax (while the opposite is naturally true for the running time). 
The heuristics again clearly outperform Random, with SA sensing approximately $15$ more targets. Random 2 performs similarly to the suboptimal heuristic SA\_Relax, while being slower by a factor of more than $100$.

Finally, we are interested in exploring the power of coordination for \texttt{Red}, i.e., the difference between the optimal utility \texttt{Blue} gets in the non-coordinated setting explored in this section compared to its utility in the Stackelberg equilibria from \Cref{sec:exp_bilevel}. 
We find that for the small instances where we can compute the Stackelberg equilibrium exactly \texttt{Red} can reduce \texttt{Blue}'s utility by $10\%$ to $20\%$ through coordination. 
For larger instance sizes, we no longer know the optimal solutions, which is why we resort to comparing the results of the respective SA heuristics. 
We find that for larger instances, the gap decreases with \texttt{Red} being only able to decrease \texttt{Blue}'s utility by $5\%$ through coordination in the instances from the {\bfseries\scshape Default} setting underlying  \Cref{tab:bilevel_approx2}. 
In our full version \cite{us}, we show that when \texttt{Red}'s sensors are capable of sensing more targets, coordination is more important sometimes leading to halving \texttt{Blue}'s utility.

\section{Conclusion}
By introducing Espace Sensing Games, we initiated the study of a new class of games concerned with target arrangement and motivated by security applications.  We showed that while the worst-case computational complexity of ESGs is prohibitive, our presented algorithms still have a good performance in experiments. 

There are multiple directions for future work emanating from our work.
First, pinpointing the precise complexity of computing Stackelberg equilibria remains a concrete open question. 
Second, there are other variants of ESGs beyond those studied by us. For instance, it would be possible to merge the settings studied in \Cref{sec:ESG,non-coor} into a game where sensors act greedily but \texttt{Red} can control the ordering of the sensors. In this game variant where both \texttt{Red} and \texttt{Blue} need to pick orders, it would also be possible to study simultaneous play or Stackelberg equilibria where \texttt{Red} moves first. 
Lastly, there are various other target arrangement problems to be studied. One example could be a game where \texttt{Blue} needs to place targets on a grid and \texttt{Red} cannot sense any two targets placed close to each~other.

\begin{ack}
This work was supported by the Office of Naval Research (ONR) under Grant Number N00014-23-1-2802. The views and conclusions contained in this document are those of the authors and should not be interpreted as necessarily representing the official policies, either expressed or implied, of the Office of Naval Research or the U.S. Government.
\end{ack}

\bibliographystyle{plain}

\clearpage

\appendix

\section{Additional Material for \Cref{sec:compRed}}

\redresponseDP*
\begin{proof}
For our dynamic program, we create a table $J[i,b_1,\dots,b_{\tau+1}]\in \mathbb{N}$ for each $i\in [n]$ and $b_1,\dots,b_{\tau+1}\in S\cup \{\emptyset\}$. 
$J[i,b_1,\dots,b_{\tau+1}]$ stores the minimum value of surviving target from $t_{1},\dots, t_{i}$ induced by a valid sensing plan $\psi$ that for each $j\in [\max(1,i-\tau),i]$ assigns target $t_{j}$ to sensor $b_{i-j+1}$  if $b_j\neq \emptyset$ and to no sensor if $b_j=\emptyset$. 
In this case, we say that $\psi$ witnesses the table entry.
If no such sensing plan exists, $J[i,b_1,\dots,b_{\tau+1}]$ is $\infty$. 
Clearly, 
the answer to the problem is $\min_{b_1,\dots,b_{\tau+1}\in S\cup \{\emptyset\}}J[n,b_1,\dots,b_{\tau+1}]$. 

For the initialization, we set an entry $J[i,b_1,\dots,b_{\tau+1}]$  to $0$ if $i=0$ and to $\infty$ otherwise. 
Now, we update the table for increasing $i=0,1,\dots,n$ by filling $J[i,b_1,\dots,b_{\tau+1}]$ as follows. We start by assuming that $b_1\neq \emptyset$ implying that $t_i$ is to be sensed by $b_1$. In this case, we set the entry to $\infty$ if  $\exists \ell\in [2,\tau+1]: b_1=b_{\ell}$ (i.e., the recharging period of $b_1$ would be violated) or if $b_1$ is not capable of sensing $t_i$. 
In both cases, there is no sensing plan to witness the entry.
Otherwise, we update the entry as: 
$$J[i,b_1,\dots,b_{\tau+1}]=\min_{s\in S\cup \{\emptyset\}} J[i-1,b_2,\dots,b_{\tau+1},s].$$
We claim that in case an update of the entry happens here, it is correct.
For this, let $s^*\in S\cup \{\emptyset\}$ be the minimizer of the right hand side and $\psi^*$ the plan witnessing  $J[i-1,b_2,\dots,b_{\tau+1},s^*]$. Then, we can extend $\psi^*$ to a sensing plan $\psi'$ by adding $t_i$ to $\psi^*(b_1)$. Plan $\psi'$ is valid and a witness for the updated entry, since $\psi^*$ is valid and we ruled out above that $b_1$ is either not ready or capable of sensing $t_i$.

Lastly, it remains to consider the case $b_1=\emptyset$, where $t_i$ is not assigned to any sensor. In this case,  we update the entry as: 
$$J[i,b_1,\dots,b_{\tau+1}]=\min_{s_\in S\cup \{\emptyset\}} J[i,b_2,\dots,b_{\tau+1},s]+v_i.$$
Analogous to above, in case an update happens assume that  $s^*\in S\cup \{\emptyset\}$ is the minimizer of the right-hand side and $\psi^*$ the plan witnessing  $J[i-1,b_2,\dots,b_{\tau+1},s^*]$. 
Then $\psi^*$ is also a witness for the updated entry. 

The algorithm runs in $\mathcal{O}(n\cdot (k+1)^{\tau+2})$, as the table contains $\mathcal{O}(n\cdot (k+1)^{\tau+1})$ entries and for each entry we need to take the minimum over $k+1$ values. 

For $i\in [k_\chi]$, let $\ell_i$ be the number of sensors of type $\theta_i$.
To extend the algorithm to sensor types, we need to prove the following lemma that allows us to collapse all sensors of the same type $\theta_i$ into a ``meta''-sensor, which can sense $\ell_i$ many targets in each $\tau$-time window:
\begin{lemma}
There is a valid sensing plan of value $v$ if and only if there is a sensing plan of value $v$ where sensor capabilities are respected and for each $j\in [1,n-\tau]$ at most $\ell_i$ targets from $\{t_j,\dots, t_{j+\tau}\}$ are assigned to a target of type $\theta_i$ for $i\in [k_\chi]$.
\end{lemma}
\begin{proof}
     Let $\psi$ be a valid sensing plan of value $v$. We claim that $\psi$ also fulfills the second condition. Assume for the sake of contradiction that there is a $\theta_i$ and some $j\in [1,n-\tau]$ so that more than $\ell_i$ targets from $\{t_j,\dots, t_{j+\tau}\}$ are assigned to a target of type $\theta_i$. Then, by the pigeonhole principle, at least one sensor of type $\theta_i$ needs to sense two targets within $\tau$ steps, rendering the plan invalid. 

    For the reverse direction, assume $\psi$ is the sensing plan of value $v$ respecting the condition. For each sensor type $\theta_i$, let $T_\psi(\theta_i)$ be the targets sensed by sensors of type $\theta_i$ in $\psi$. 
    We construct a sensing plan $\psi'$ as follows. 
    For each sensor type $\theta_i$, we iterate over the targets in $T_\psi(\theta_i)$ according to their position in the target ordering and always assign a target to the sensor of type $\theta_i$ who has not sensed a target for the longest time. 
    $\psi'$ has the same value as $\psi$ and respects sensor capabilities, so it remains to argue that the recharging periods are respected. 
    For the sake of contradiction assume that there is a sensor of type $\theta_i$ and some $j$ so that the sensor is assigned two targets from $\{t_j,\dots, t_{j+\tau}\}$ in $\psi'$. 
    Then, by the construction of $\psi'$ all other sensors of type $\theta_i$ are also assigned a target from $\{t_j,\dots, t_{j+\tau}\}$. It follows that sensors of $\theta_i$ are assigned at least $\ell_i+1$ many targets from $\{t_j,\dots, t_{j+\tau}\}$ in $\psi'$ and thereby also in $\psi$, which contradicts our initial assumptions on $\psi$.     
\end{proof}
Using this lemma we can easily adjust the dynamic programming formation: Instead of bookmarking the sensors $b_1,\dots,b_{\tau+1}\in S$ that have sensed the last $\tau+1$ targets, we instead bookmark the types $\theta'_1,\dots,\theta'_{\tau+1}\in \Theta$ of these sensors. Now, we can set a table entry to $\infty$ (due to violated recharging time) if $\{\theta'_1,\dots,\theta'_{\tau+1}\}$ contains more some sensor type $\theta_i\in \Theta$ more than $\ell_i$ times. The rest of the algorithm adapts in a straightforward manner with a resulting running time of $\mathcal{O}(n\cdot (k_\chi +1)^{\tau+2})$
\end{proof}

\section{Additional Material for \Cref{sec:Stackle}}\label{app:Stackle}
\stackcomp*
\begin{proof}
Let $\sigma:T\to [n]$ be some target ordering. 
We will say that two targets $t,t'\in T$ are (placed) at distance $j$ (in $\sigma$) if $|\sigma(t)-\sigma(t')|=j$. 
Further, for a target $t$ we define its $i$-surrounding to be the set of all targets whose distance from $t$ is at most $i$. 

We start the proof with two immediate claims about optimal sensing plans in response to some target ordering $\sigma$:
\begin{claim}\label{c:1}
    Let $s$ be a sensor and $t$ be a target so that $s$ is the only sensor capable of sensing $t$ and $t$ has a higher value than all the other targets $s$ can sense combined.  $s$ senses $t$ in an optimal sensing plan. 
\end{claim}
\begin{proof}
Assume that $s$ did not sense $t$, we could just arrive at a better valid sensing plan by letting $s$ sense only $t$. 
\end{proof}

\begin{claim}\label{c:2}
    Let $s$ be a sensor with recharging time $\tau$ and $t$ be some target $s$ is capable of sensing. If there is no other target $s$ capable of sensing in the $\tau$-surrounding of $t$, then $t$ will be sensed by some sensor in the optimal sensing plan.
\end{claim}
\begin{proof}
    Assume that $t$ is not sensed by any sensor. Then we can just arrive at a better valid sensing plan by letting $s$ sense $t$: This will not violate the recharging constraint, as $s$ does not sense any other targets in the $\tau$-surrounding of $t$. 
\end{proof}

We reduce from \textsc{Independent Set} on $3$-regular triangle-free graphs \cite[Theorem 3]{DBLP:conf/wg/GuruswamiRCCW98}.

\paragraph{Construction}
Let $G=(V,E)$ be a $3$-regular triangle-free graph.
For each vertex $v\in V$, we introduce a \emph{vertex target} $\alpha_v$ of value $1$ and for each $i\in [3]$ a \emph{blocker vertex target} $\beta^i_v$ of value $3$ and a \emph{vertex sensor} $a^i_v$ which is capable of sensing $\alpha_v$ and $\beta^i_v$. 
Moreover, for each edge $e=\{u,v\}\in E$, we introduce an \emph{edge target} $\gamma_e$ of value $3$ and an \emph{edge sensor}  $g_e$ that can sense $\gamma_e$, $\alpha_u$, and $\alpha_v$. 
Additionally, we add a \emph{constraining edge target} $\hat{\gamma}_e$ of value $4$ and a \emph{constraining edge sensor} $\hat{g}_e$  that is capable of sensing targets $\gamma_e$ and $\hat{\gamma}_e$.  
The recharging time is $\tau=3$. 
We claim that there is an $\ell$-sized independent set  $X$ in $G$ if and only if the Stackelberg equilibrium in the constructed ESG has value at least $\ell$. 

\paragraph{Forward Direction}
Assume that $X\subseteq V$ is an independent set of size $\ell$. 
To construct the target ordering we go through the vertices from the independent set one by one:  
For $v\in X$, let $e_1,e_2,e_3$ be the three edges incident to $v$. 
Then, we add to the ordering of targets the targets $\hat{\gamma}_{e_1}\beta^1_v,\beta^2_v,\gamma_{e_1},\alpha_v,\gamma_{e_2},\gamma_{e_3},\beta^3_v,\hat{\gamma}_{e_2},\hat{\gamma}_{e_3}$ in this order. 
After we have processed all vertices from the independent set like this, we append all other targets in a random ordering. 
First observe that because of  \Cref{c:1}, the constraining edge sensors will always sense the constraining edge targets. 
As a result of the structure of our target ordering, the constraint edge sensors for edges incident to $v\in X$ cannot sense the corresponding edge targets. 
From this, following a reasoning analogous to \Cref{c:1}, we get that for each edge $e$ incident to a vertex from $X$, the edge sensor $g_e$ senses the edge target $\gamma_e$.
Moreover, we also get from \Cref{c:1} that all three vertex sensors corresponding to a vertex $v\in X$ sense their respective blocker vertex target. 
All in all, it follows that for each vertex $v\in X$ all three corresponding vertex sensors and all edge sensors corresponding to incident edges sense a target within distance $3$ of $\alpha_v$. Thus, none of the sensors that are capable of sensing $v$ can do so without violating their recharging constraint.
Thus, all $\ell$ vertex targets corresponding to vertices from $X$ make it unsensed through the channel.

\paragraph{Backward Direction}
Assume that we have a target ordering $\sigma$ so that targets of summed value at least $\ell$ make it unsensed through the channel under optimal play by \texttt{Red}. 
Let $\psi$ be the optimal sensing plan played in response by \texttt{Red}. 
From \Cref{c:1}, it is immediate that no blocker vertex target and no constraining edge target can make it unsensed through the channel. Moreover, it is also easy to see that no edge target can make it unsensed through the channel, as otherwise we could always improve the plan by letting the corresponding edge sensor sense the edge target (and delete all vertex targets the sensor senses instead). 
Let $X$ be the set vertices so that the corresponding vertex targets make it unsensed through the channel. 
By the above observation, we have $|X|\geq \ell$. We will now show that $X$ is an independent set. 

For this, we make a series of observations: 
First, let $\alpha_v$ be an unsensed vertex target. 
Note that $\alpha_v$ can be sensed by $6$ different sensors (three vertex sensors and three edge sensors). 
By \Cref{c:2}, we have that for all of these sensors there needs to be another target that the sensor is capable of sensing in the $3$-surrounding of $\alpha_v$. 
As there are $6$ targets in a $3$-surrounding, it follows that the $3$-surrounding of $\alpha_v$ contains only targets that can be sensed by one of these six sensors (specifically one target for each of these six sensors). 

Second, we show that for each unsensed vertex target $\alpha_v$ there is no other vertex target in its $3$-surrounding. 
Recall from the first observations that the only vertex targets that could be in the $3$-surrounding of $\alpha_v$ are vertex targets corresponding to neighbors of $v$ that are sensed by the corresponding edge sensor in $\psi$. 
For the sake of contradiction assume that there is some $u\in V$ so that $e=\{u,v\}\in E$ and $|\sigma(\alpha_v)-\sigma(\alpha_u)|\leq 3$. 
Observe that the $3$-surroundings of $\alpha_u$ and $\alpha_v$ overlap in at least $3$ targets. 
As $G$ is triangle-free and thus $u$ and $v$ do not share any common neighbors,
this implies that the $3$-surrounding of $\alpha_u$ contains at most three targets that can be sensed by one of the six sensors that are capable of sensing $\alpha_u$ (the other three spots will be filled with vertex blocker targets corresponding to $v$ and edge targets corresponding to edges incident to $v$ by the first observation and as $v$ makes it unsensed through the channel). 
Let $s$ be a sensor that is capable of sensing $\alpha_u$ for which no other target that $s$ can sense in the $3$-surrounding of $\alpha_u$. 
Next, note that by the first observation, $\alpha_u$ is sensed by $g_e$ in $\psi$. 
We alter the sensing plan $\psi$ as follows. We let $\alpha_v$ instead of $\alpha_u$ be sensed by $g_e$ and let $\alpha_u$ be sensed by $s$. 
As argued above, this does not violate the recharging time of the sensor $s$. 
Moreover, it also does not violate the recharging time of $g_e$, as there are no other targets (except $\alpha_u$) that $g_e$ is capable of sensing in the $3$-surrounding of $\alpha_v$ (by the first observation). 
In the altered sensing plan the value of surviving targets is strictly smaller contradicting that $\psi$ is an optimal response by \texttt{Red}. 

Third, we claim that for each of two unsensed vertex targets $\alpha_v$ and $\alpha_{v'}$ we have that they are placed at a distance at least $7$ from each other, i.e., their $3$-surroundings do not overlap. 
In order to show that this cannot be the case, we need to examine the constraining edge sensors. 
Let $e$ be an edge incident to $v$.
Our first two observations told us that the edge target $\gamma_{e}$ needs to be in the $3$-surrounding of $\alpha_v$. 
However, this is not sufficient:
We claim that the constraining edge target $\hat{\gamma}_{e}$ needs to be in the $3$-surroundings of $\gamma_{e}$. 
Assume that this was not to hold. Then we could alter our sensing plan $\psi$ by making $\hat{g}_e$ sense $\gamma_e$ (it can do so without violating its recharging time because $\hat{\gamma}_{e}$ is not in the $3$-surrounding of $\gamma_{e})$ and make $g_e$ sense $\alpha_v$, which leads to a strictly better sensing plan. The claim follows.
To prove the observation, for the sake of contradiction assume that the $3$-surroundings of $\alpha_v$ and $\alpha_{v'}$ did overlap.
From the previous two observations, it follows that the only possibility for their $3$-surrounding to overlap is that there is an edge $e=\{v,v'\}\in E$ and that $\gamma_e$ constitutes the overlap of their $3$-surrounding. 
However, in this case $\hat{\gamma}_{e}$ cannot be placed in the $3$-surrounding of $\gamma_e$, as it can neither belong to the $3$-surrounding of $\alpha_v$ nor $\alpha_{v'}$ (by the first observation). This leads to a contradiction to our above claim. 

Now, combining these observations it follows that for each $v\in X$, the $3$-surrounding of $\alpha_v$ contains edge targets $\gamma_{e_1},\gamma_{e_2},\gamma_{e_3}$ with $e_1$, $e_2$, and $e_3$ being the edges incident to $v$. 
Moreover, we have shown that the $3$-surrounding of each target $\alpha_v$ with $v\in X$ is disjoint. 
This implies that no two vertices $v,u\in X$ can be incident to the same edge $e$, as otherwise, $\gamma_e$ would be in the $3$-surrounding of $\alpha_u$ and $\alpha_v$, which leads to a contradiction as they are disjoint. 
It follows that $X$ is an independent set of size at least $\ell$.

\end{proof}

\begin{algorithm}[h]
\caption{\textsc{Greedy\_Sensing}}
\label{alg:Greedy_sensing}
\textbf{Input}: Blue team's ordering $\sigma$ of targets \\
\textbf{Output}: Red team's sensing plan $\psi$ 
\begin{algorithmic}[1] 
\STATE Sort \texttt{Blue}'s target according to their value \\
\STATE Label \texttt{Red}'s sensors with $\{1, \cdots, k\}$
\FOR{each target in the order of decreasing values}
    \STATE Find all sensors that are available to sense the target with index $i$ in the input ordering
    \IF{Only one sensor $s \in [k]$ can sense}
        \STATE Set $\psi_{i,s} =1$
    \ELSIF{Multiple sensors $S \subseteq [k]$ can sense}
        \STATE Choose one $s \in S$ randomly and set $\psi_{i,s} =1$
    \ENDIF
\ENDFOR
\STATE \textbf{Return} $\psi$
\end{algorithmic}
\end{algorithm}

\bilevel*
\begin{proof}
We start by giving the formulation, where the inner-level program is similar to the ILP presented in \Cref{pr:ILP-best}. Note that the outer-level maximizes the value of non-sensed targets, while the inner-level minimizes this value:
\begin{align}
    \max_{z, t} &\sum_{i\in [n]} v_i \cdot x_{i,k+1} \label{e1} \\
    \text{s.t. } &z_i \ne z_{i'},  \forall i\ne i' \in [n] \label{e2} \\ 
    &\tau + 1 - n\cdot(1-l_{i,j})\le z_j - z_i,  \forall i,j\in[n] \label{e4} \\
    &z_j - z_i \le n\cdot l_{i,j} + \tau,  \forall i,j\in[n] \label{e3} \\
    &\mathbf{x} = \text{argmin}_{\hat{\mathbf{x}}} \sum_{i\in [n]} v_i \cdot \hat{x}_{i,k+1} \label{e4} \\
    \text{s.t. } &\sum_{j\in [k+1]} \hat{x}_{i,j}=1, \forall i \in [n]\label{e5} \\
    &\hat{x}_{i,j} \le D_{i,j}, \forall i \in [n], j \in [k] \label{e6} \\
    &\hat{x}_{i,j}+\hat{x}_{i',j}\leq 1+l_{i,i'}+l_{i',i}, \forall i\ne i' \in [n], j \in [k] \label{e7}
\end{align}
The outer-level program controls an integer variable $z_i$ for each target $i\in [n]$ that encodes the position in which the target appears in the final ordering. 
Moreover, for each pair of targets $i,j\in [n]$, we have a binary variable $l_{i,j}$ capturing whether $t_i$ appears at least $\tau+1$ positions before $t_j$ in the ordering induced by the $z$ variables. 
\Cref{e2} ensures that each target is assigned a unique position\footnote{For each $i\neq i'\in [n]$, to convert \Cref{e2} into a linear constraint, we have to introduce a new binary variable $\delta_{i,i'}$. We then add two linear constraint $z_i \le z_j - 1 + n\cdot \delta_{i,j}$ and $z_i \ge z_j + 1 + n\cdot (1-\delta_{i,j})$ that are satisfied if and onlf if $z_i\neq z_{i'}$. \label{foot}} and \Cref{e3,e4} ensures that $l_{i,j}$ is set to one if and only if  $z_j\geq z_i+\tau +1$ (in \Cref{e3,e4}  we have $\tau + 1 \le z_j - z_i$ if $l_{i,j} = 1$; and $z_j - z_i \le \tau$ if $l_{i,j}=0$).

The inner-level program controls binary variables $x_{i,j}$ for each $i\in [n]$ and $j\in [k+1]$, which encode the sensing plan as in \Cref{pr:ILP-best}, i.e, $x_{i,j}=1$ for $i\in [n]$ and $j\in [k]$ implies that $t_i$ is sensed by $s_j$. 
The value of the encoded plan is again $\sum_{i\in [n]} v_i\cdot x_{i,k+1}$, which \texttt{Blue} wants to maximize (\Cref{e1}) and \texttt{Red} wants to minimize (\Cref{e4}). 
The validity of the sensing plan is secured in \Cref{e5,e6,e7}, where \Cref{e7} imposes that a sensor $s_j$ can only sense two targets $t_i$ and $t_{i'}$ if either $t_i$ is at least $\tau+1$ positions before $t_{i'}$ in the ordering encoded in the $z$ variables (i.e., $l_{i,i'}=1$) or the other way around (i.e., $l_{i',i}=1$).
\end{proof}

We give our greedy approximation algorithm for \textsc{Best Red Response} in \Cref{alg:Greedy_sensing}. 

\section{Additional Material for \Cref{sec:NonCoordAA}}

\subsection{Proof of \Cref{th:bestblue}}

\subsubsection{Connection \textsc{Best Blue Response} and \textsc{Minimum Maximal Matching}} \label{sec:connection}
On an intuitive level, solving \textsc{Best Blue Response} with infinite recharging time $\tau=\infty$ and uniform target values has some similarities to solving the classic NP-hard \textsc{Minimum Maximal Matching} problem Yannakakis and Gavril \cite{yannakakis1980edge}: 
\decprob{Minimum Maximal Matching}{A bipartite graph $G=(U\cup V, E)$ and an integer $\ell$.}{Is there a maximal matching in $G$ containing at most $\ell$ edges?}

We now discuss the intuitive connection as well as reasons why immediate reductions between the two problems are prohibited. 
Assume that we have a solution to our \textsc{Best Blue Response} instance where targets are ordered as $(t_{1}, \dots, t_{n})$. 
From this let us construct a bipartite graph $G$ with vertices $T$ on the left side and vertices $S$ on the right side. For the edge set $E$, we add an edge between a target $t$ and a sensor $s$ if $s$ is capable of sensing $t$
Let now $F\subseteq E$ be the set of sensor-target pairs with $\{t,s\}\in F$ if $s$ senses $t$ when targets are send according to $(t_{1}, \dots, t_{n})$. 
It needs to hold that $F$ is a maximal matching in $G$: 
Otherwise, there is some $\{t,s\}\in E\setminus F$. 
The existence of this edge implies that target $t$ made it through the channel and sensor $s$ did not sense any target, which leads to a contradiction. 
Moreover, the size of this matching, i.e., $|F|$, corresponds to the number of lost targets.
Thus, \texttt{Blue} wants to find a maximal matching of minimum size. 
This discussion suggests a close connection between \textsc{Simple Sequential Covering} and \textsc{Minimum Maximal Matching}. 

However, there are some crucial differences between the two problems which prohibit immediate reductions from one problem to the other: 
Most crucially, assume we were to model a bipartite graph as an instance of \textsc{Best Blue Response} by letting $U$ be the targets and $V$ the sensors (in some ordering $(v_1,\dots, v_n)$) and  let $v$ be capable of sensing $u$ if $\{u,v\}\in E$. 
The problem with this construction is that we cannot model arbitrary matchings $E'\subseteq E$ as solutions to the constructed \textsc{Best Blue Response} instance: 
Assume that $E'$ contains some edge $\{u,v_i\}$ and there is some $v_j$ with $j<i$, $\{u,v_j\}\in E$, and $v_j$ is not incident to any edges from $E'$. 
In this case, it is not possible to send the targets through the channel such that $E'$ will be the sensed sensor-target pairs because $u$ will always be sensed by $v_j$ before it can be sensed by $v_i$ (which in turn implies that $v_i$ is still ready to sense other targets). 
As a consequence, intuitively speaking, in instances of \textsc{Best Blue Response} we are only interested in maximal matchings where each matched vertex from the left is matched to its ``first'' otherwise unmatched neighbor from the right side. 

Because of this, we need to turn to a slightly more involved reduction that draws inspiration from the NP-hardness proof of \textsc{Minimum Maximal Matching} by Yannakakis and Gavril \cite{yannakakis1980edge}, yet requires some reworking of the construction and a different more involved proof. 
We reduce from the following \textsc{SAT} variant.

\subsubsection{Proof of Correctness}
\bestblue*
\begin{proof}
In this proof, for a target $t$, we let $D(t)$ be the set of sensors that are capable of sensing $t$. 

We reduce from the following problem, which is  NP-hard as proven by Yannakakis \cite{DBLP:journals/siamcomp/Yannakakis81}. 
\decprob{Restricted 3-Sat}{A propositional formula $(X,C)$ where each clause contains three literals and each variable appears in exactly two clauses positively and in exactly one clause negatively.}{Is there an assignment to variables in $X$ such that each clause from $C$ is satisfied?}

Let $(X,C)$ be a given \textsc{Restricted 3-Sat} instance. 

\paragraph{Construction.}
Each target has value one.
For each clause $c\in C$, we add a \emph{clause} target $t_c$ and a \emph{clause} sensor $s_c$ with $D({t_c})=\{s_c\}$.

For each variable $x\in X$, we add a variable gadget.
That is, we add \emph{variable} targets $t_{x,1}$, $t_{x,2}$ and $t_{\bar{x}}$ together with \emph{dummy} targets $t^{\mathrm{du}}_{x,1}$, $t^{\mathrm{du}}_{x,2}$, $t^{\mathrm{du}}_{x,3}$, and $t^{\mathrm{du}}_{x,4}$.
Next, we add \emph{catch} sensors $s^{\mathrm{ca}}_{x,1}$, $s^{\mathrm{ca}}_{x,2}$, and $s^{\mathrm{ca}}_{\bar{x}}$, which ensure that none of the variable targets can make it through the channel. 
Moreover, we add \emph{variable} sensors $s_{x,1}$, $s_{x,2}$ and $s_{\bar{x}}$ together with \emph{dummy} sensors $s^{\mathrm{du}}_{x,1}$ and $s^{\mathrm{du}}_{x,2}$.

Let $c_i,c_j,c_k\in C$ be three clauses such that $x$ appears positive in clauses $c_i$ and $c_j$ and negative in clause $c_k$. 
The sensing matrix is defined through:
\begin{align*}
   & D({t_{x,1}})=\{s_{c_i},s_{x,1},s^{\mathrm{ca}}_{x,1}\}   & D({t_{x,1}^{\mathrm{du}}})=\{s_{x,1},s^{\mathrm{du}}_{x,1}\} \\
   & D({t_{x,2}})=\{s_{c_j},s_{x,2},s^{\mathrm{ca}}_{x,2}\} & D({t_{x,2}^{\mathrm{du}}})=\{s_{\bar x},s^{\mathrm{du}}_{x,1}\}  \\
   & D({t_{\bar x}})=\{s_{c_k},s_{\bar x},s^{\mathrm{ca}}_{\bar x}\} & D({t_{x,3}^{\mathrm{du}}})=\{s_{\bar x},s^{\mathrm{du}}_{x,2}\} \\
   & &   D({t_{x,4}^{\mathrm{du}}})=\{s_{x,2},s^{\mathrm{du}}_{x,2}\}\\
\end{align*}

The recharging time is $\infty$, i.e., each sensor can sense at most one target.
The ordering of the sensors is as follows. 
First, come the clause sensors (in some arbitrary ordering), then the dummy sensors (in some arbitrary ordering), then the variable sensors (in some arbitrary ordering) and last the catch sensors (in some arbitrary ordering).  
We ask whether there is an ordering of targets so that at least $\ell:=|C|+2|X|$ targets are not sensed.
It is easy to see that the construction satisfies the restrictions from the theorem statement.

\paragraph{Proof of Correctness: Forward Direction}
Assume we are given an assignment of variables $X$ that fulfills $C$. 
Let $X^*\subseteq X$ be the set of variables set to true in this assignment. 
From this, we construct a partition of the targets into four groups that determine the ordering in which the targets move through the channel; the first group comes first and so on;  the ordering of targets within one group is arbitrary. 
The first group consists of  targets $t_{x,1}$ and $t_{x,2}$ for each $x\in X^*$ and 
target $t_{\bar x}$ for each $x\notin X^*$. 
The second group consists of targets $t^{\mathrm{du}}_{x,1}$ and $t^{\mathrm{du}}_{x,4}$ for each $x\in X^*$ and $t^{\mathrm{du}}_{x,2}$ and $t^{\mathrm{du}}_{x,3}$  for each $x\notin X^*$. 
The third group contains all remaining variable targets. 
And, finally, the fourth group contains all $|C|$ clause targets and the remaining $2|X|$ dummy targets. 
It is sufficient to prove that all $\ell$ targets from the fourth group make it through the channel. 
We prove this via a series of three claims. 

\begin{claim}
\begin{enumerate}
    \item Each clause sensor senses a target from the first group. 
    \item Each dummy sensor senses a target from the second group. 
    \item For each $x\in X^*$, $s_{\bar x}$ senses a target from one of the first three groups, and for each $x\notin X^*$, $s_{x,1}$ and  $s_{x,2}$  sense a target from one of the first three groups.
\end{enumerate}
\end{claim}
\begin{proof}
\textbf{Proof of 1.} This follows directly from the fact that $X^*$ is a satisfying assignment and that the clause sensors come first in the ordering of sensors. 

\textbf{Proof of 2.} This follows directly from the fact that the dummy sensors come after the clause sensors in the sensor ordering and as the second group contains for each $x\in X$ $t^{\mathrm{du}}_{x,1}$ or $t^{\mathrm{du}}_{x,2}$ (making $s^{\mathrm{du}}_{x,1}$ sense a target) and  $t^{\mathrm{du}}_{x,3}$ or $t^{\mathrm{du}}_{x,4}$ (making $s^{\mathrm{du}}_{x,2}$ sense a target).

\textbf{Proof of 3.} Let us focus on one $x\in X$, and let $c_i,c_j,c_k\in C$ be three clauses such that $x$ appears positive in clauses $c_i$ and $c_j$ and negative in clause $c_k$. 
If $x\in X^*$, then $t_{\bar x}$ is part of the third group. However, from Statement 1 it follows that $s_{c_k}$ already sensed a previous target. 
As the variable sensors are before the catch sensors in the sensor order, it follows that $s_{\bar x}$ senses $t_{\bar x}$. 
Similarly, if $x\notin X^*$, $t_{x,1}$ and $t_{x,2}$ are part of the third group. Both $s_{c_i}$ and $s_{c_j}$ have already sensed previous targets because of Statement 1. From this Statement 3 follows. 
\end{proof}
The claim implies that all sensors that can sense a target from the fourth group have already sensed another target before it is the fourth group's turn, implying that all $\ell$ targets from the fourth group will make it through the channel. 

\paragraph{Proof of Correctness: Backward Direction}
Assume that there is an ordering of the targets such that at least $\ell$ targets move unsensed through the channel, and let $P^*\subseteq P$ be the set of these targets.

We prove the backward direction in a series of claims: 
\begin{claim}\label{claim:back}
    \begin{enumerate}
        \item No variable target is part of $P^*$. 
        \item For each $x\in X$, either $t^{\mathrm{du}}_{x,1}$ or $t^{\mathrm{du}}_{x,2}$ and either $t^{\mathrm{du}}_{x,3}$ or $t^{\mathrm{du}}_{x,4}$ is part of $P^*$. All clause targets are part of $P^*$.
        \item For each $x\in X$, if $t_{\bar x}$ is sensed by a clause target, then neither $t_{x,1}$ nor $t_{x,2}$ are sensed by a clause target. 
        \end{enumerate}
\end{claim}

\begin{proof}
\textbf{Proof of 1.} This follows immediately from the existence of a designated catch sensor for each variable target which can only sense this target. As a consequence, no variable target can ever make it unsensed through the channel. 

\textbf{Proof of 2.} Note that because of the sensor $s_{x,1}^{\mathrm{du}}$ it is never possible that both $t^{\mathrm{du}}_{x,1}$ and $t^{\mathrm{du}}_{x,2}$ make it unsensed through the targets. Similarly, because of $s_{x,2}^{\mathrm{du}}$ it is never possible that both $t^{\mathrm{du}}_{x,3}$ and $t^{\mathrm{du}}_{x,4}$ make it unsensed  through the channel. 
Together with Statement 1, this implies that from each variable gadget at most $2$ targets can be part of $P^*$. 
By recalling that $\ell=|C|+2\cdot |X|$ and that there are only $|C|$ clause targets outside of variable gadgets, the statement follows. 

\textbf{Proof of 3.} Let us focus on $t_{x,1}$ (the proof for $t_{x,2}$ is analogous). 
For the sake of contradiction assume that $t_{x,1}$ and $t_{\bar x}$ are both sensed by clause sensors, then both $s_{x,1}$ and $s_{\bar x}$ do not sense a variable target, respectively. 
Accordingly, at most one target out of $\{t^{\mathrm{du}}_{x,1},t^{\mathrm{du}}_{x,2},t^{\mathrm{du}}_{x,3},t^{\mathrm{du}}_{x,4}\}$ can make it unsensed through the channel, contradicting Statement 2.
\end{proof}
Let $\alpha$ be a truth assignment that sets $x\in X$ to false if $t_{\bar x}$ is sensed by a clause sensor and $x\in X$ to true if $t_{x,1}$ or $t_{x,2}$ is sensed by a clause sensor.
If neither of the two conditions hold, then we set $x$ to true.
Note that the well-definedness of $\alpha$ follows immediately from Statement 3 of \Cref{claim:back}.
Assume that $c\in C$ is not satisfied by $\alpha$. However, this implies that $s_{c}$ does not sense a variable target corresponding to a literal appearing in $c$ (by the definition of $\alpha$).
This implies that $s_c$ will sense $t_c$, a contradiction to Statement 2 of \Cref{claim:back}. 
\end{proof}

\subsection{Proof of \Cref{bbr:DP}}
\DBGreedy*
\begin{proof}
Recall that we assume that sensors act greedily and that the sensor ordering is fixed and known.
We iteratively construct the target ordering always appending an additional target at the end of the ordering, while storing the types of already sent targets as well as the sensors that sensed the last $\tau+1$ targets.
For our dynamic program, we create a table $J[i_1,\dots,i_{n_\chi},b_1,\dots,b_{\tau+1}]\in \mathbb{N}$ with $i_j\in [\ell_j]$ for each $j\in [n_\chi]$ and $b_1,\dots,b_{\tau+1}\in S\cup \{\emptyset\}$. 
For a table cell, let $i:= \sum_{j\in [n_\chi]} i_j$, i.e., $i$ is the total number of targets that have been sent. 
An entry of the table stores the maximum value of targets that can survive if \texttt{Blue} sends $i_j$ targets of type $\gamma_j$ (for each $j\in [n_\chi]$) through the channel in a way that the $t$th last target sent for $t\in [\max(1,(\sum_{j\in [n_\chi]} i_j)-\tau),\sum_{j\in [n_\chi]} i_j]$ is sensed by sensor 
if $b_j\neq \emptyset$ and by no sensor if $b_j=\emptyset$ (the intuition is that $b_1$ is the sensor that sensed the most recent passing target (if existent), $b_2$ sensed the second most recent one, and so on. 
If the second constraint is not realizable, we set the table entry to $-\infty$. 
The answer to our problem is $\min_{b_1,\dots,b_{\tau+1}\in S\cup \{\emptyset\}}J[\ell_1,\dots,\ell_{n_\chi},b_1,\dots,b_{\tau+1}]$. 

For the initialization, we set an entry $J[i_1,\dots,i_{n_\chi},b_1,\dots,b_{\tau+1}]$  to $0$ if $\sum_{j\in [n_\chi]} i_j=0$ and to $-\infty$ otherwise. 
Now, we update the table for increasing $\sum_{j\in [n_\chi]} i_j=0,1,\dots,n$ by filling $J[i_1,\dots,i_{n_\chi},b_1,\dots,b_{\tau+1}]$ as follows.
We start by assuming that $b_1\neq \emptyset$ implying that the next target to be sent needs to be sensed by $b_1$.
If $b_1$ appears among $b_2,\dots, b_{\tau+1}$, we set the entry to $-\infty$. Otherwise, let $\gamma_{j_1}, \dots, \gamma_{j_z}$ be the target types so that $b_1$ is capable of sensing targets of this type and only sensors from $\{b_2,\dots,b_{\tau+1}\}$ appear before $b_1$ in the sensor ordering and are capable of sensing targets of this type. Less formally speaking, $\gamma_{j_1}, \dots, \gamma_{j_z}$  are all the target types so that if a target of this type is sent next over the channel $b_1$ would be the sensor sensing this target (as all other sensors placed before $b_1$ that can sense targets of this type are still recharging, i.e., they are part of $\{b_2,\dots,b_{\tau+1}\}$). 
 If no such target type exists, we set  $J[i_1,\dots,i_{n_\chi},b_1,\dots,b_{\tau+1}]$ to $-\infty$.
 Otherwise, for each $t\in [z]$, we check $\max_{s\in S\cup \{\emptyset\}} J[i_1,\dots,i_{\gamma_{j_t}}-1,\dots, i_{n_\chi},b_2,\dots,b_{\tau+1},s]$ and let the entry be the maximum of these values. 

Analogously, if $b_1= \emptyset$, we let 
$\gamma_{j_1}, \dots, \gamma_{j_z}$ be the target types where $\{b_2,\dots,b_{\tau+1}\}$ contains all the sensors that are capable of sensing targets of this type. 
If no such target type exists, we set the entry to be $-\infty$. 
Otherwise,  for each $t\in [z]$, we compute $\max_{s\in S\cup \{\emptyset\}} J[i_1,\dots,i_{\gamma_{j_t}}-1,\dots, i_{n_\chi},b_2,\dots,b_{\tau+1},s]$ plus the utility of targets of type $\gamma_{j_t}$ and let the entry be the maximum of these values. 

The correctness follows from the fact that the initially stated invariant is preserved throughout the algorithm. 
Observing that computing each table entry takes $\mathcal{O}(n_{\chi}\cdot (k+1))$ time, the claimed running time of $\mathcal{O}\left(n_\chi\cdot\left(\prod_{i=1}^{n_\chi} (\ell_i+1)\right)\cdot (k+1)^{\tau+2}\right)$ follows. 

\end{proof}

\subsection{Proof of \Cref{pr:bestblueILP}}
\bestblueILP*
\begin{proof}
We model an instance $\mathcal{I}$ of \textsc{Best Blue Response} as follows.
For each target $i\in [n]$, we add an integer variable $z_i$ that encodes the position in which the target appears in the final ordering. 
We add linear constraints so that $z_i\in [1,n]$ and $z_i\neq z_{i'}$ for all $i\neq i'\in [1,n]$ (see \Cref{foot}).

Next, similar as in \Cref{pr:ILP-best}, for each $i\in [n]$ and $j\in [k+1]$, we add a binary variable $x_{i,j}$. Setting $x_{i,j}$ to one means that $t_i$ is detected by sensor $s_j$ or in case that $j=k+1$ that the target makes it unsensed through the channel. 
Accordingly, the objective becomes: 
$$\max \sum_{i\in [n]} v_i \cdot x_{i,k+1}.$$
For each target $i\in [n]$, we impose that: 
$$\sum_{j\in [k+1]} x_{i,j}=1.$$
Moreover, we impose $x_{i,j}\leq D_{i,j}$ for each $i\in [n]$ and $j\in [k]$, enforcing the sensor capabilities. 

To ensure that the recharging times of sensors are respected we add the following set of constraints. 
For each, $j\in [k]$ and $i,i'\in [n]$, we add 
$$|z_{i'}-z_{i}| \geq -n(2-x_{i,j}-x_{i',j}) + \tau. $$
This ensures that if $x_{i,j}=1$ and $x_{i',j}=1$, then $i$ and $i'$ are placed far enough away from each other, while otherwise the condition is vacant. 
To realize the absolute value from the above equation, we have to introduce another set of binary variables $o_{i,i'}$ for $i,i'\in [n]$ and add the constraints: 
$z_{i'}-z_{i}+n\cdot o_{i,i'}  \geq -n(2-x_{i,j}-x_{i',j}) + \tau $ and $z_{i}-z_{i'}+n\cdot (1-o_{i,i})  \geq -n(2-x_{i,j}-x_{i',j}) + \tau $. 

If a target $i$ is sensed by sensor $j$ then due to the recharging time sensor $j$ will not be able to sense other sensors, thus $i$ ``protects'' some targets from being sensed by sensor $j$. 
To capture this information, for each $i,i'\in [n]$ and $j\in [k]$, we add a binary variable $y_{i,i',j}$ that is equal to one target $i$ is sensed by sensor $j$ and because of this $j$ cannot sense $i'$. 
To ensure this, first, for each $i\in [n]$ and $j\in [k]$, we add: 
$$\sum_{i'\in [n]} y_{i,i',j} \leq n\cdot x_{i,j}$$
(a target $i$ can only protect other targets if the target is sensed by the corresponding sensor).
Moreover, for each $i,i'\in [n]$ and $j\in [k]$, we need to make sure that if $y_{i,i',j}=1$, then $0\leq z_{i'}-z_i \leq \tau$ (to exploit of the recharging constraint). For this, we add constraints:\begin{equation}\label{eq:cover_constr} -n(1-y_{i,i',j}) \le z_{i'}-z_{i} \le n(1-y_{i,i',j}) + \tau. \end{equation} 

Lastly, we need to make sure that a target will survive until step $j$ if $x_{i,j}=1$ i.e., the target needs to be covered by other targets for all sensors that are capable of sensing it placed before $j$. Note that this together with the first constraint ($\sum_{j\in [k+1]} x_{i,j}=1$) in particular implies that each target is sensed by the first sensor it passes which is not recharging, thereby successfully encoding the greedy behavior of the sensors. 
Specifically, we add the following set of constraints for each $i\in [n]$ and $j\in [k+1]$: 
\begin{align} \label{eq:prob}
    \sum_{t\in [j-1]: D_{i,t}=0} 1+ &\sum_{t\in [j-1]:D_{i,t}=1} \sum_{i'\in [n]} y_{i',i,t} -(j-1)\\ \nonumber &  \geq -n(1-\sum_{t=j}^{k+1} x_{i,t}).
\end{align}
\end{proof}

\section{Additional Experimental Results}\label{appendix_sec:exp}

To begin, we generate a Figure \ref{fig:prob_d} illustrating the average utility for \texttt{Blue} across varying probabilities of $D_{i,j} = 1$. This visualization aims to demonstrate the impact of the probability of $D_{i,j} = 1$ on \texttt{Blue}'s utility under the {\bfseries\scshape Default}  game settings.

\begin{figure}[tbh]
    \centering
    \includegraphics[width=0.4\textwidth]{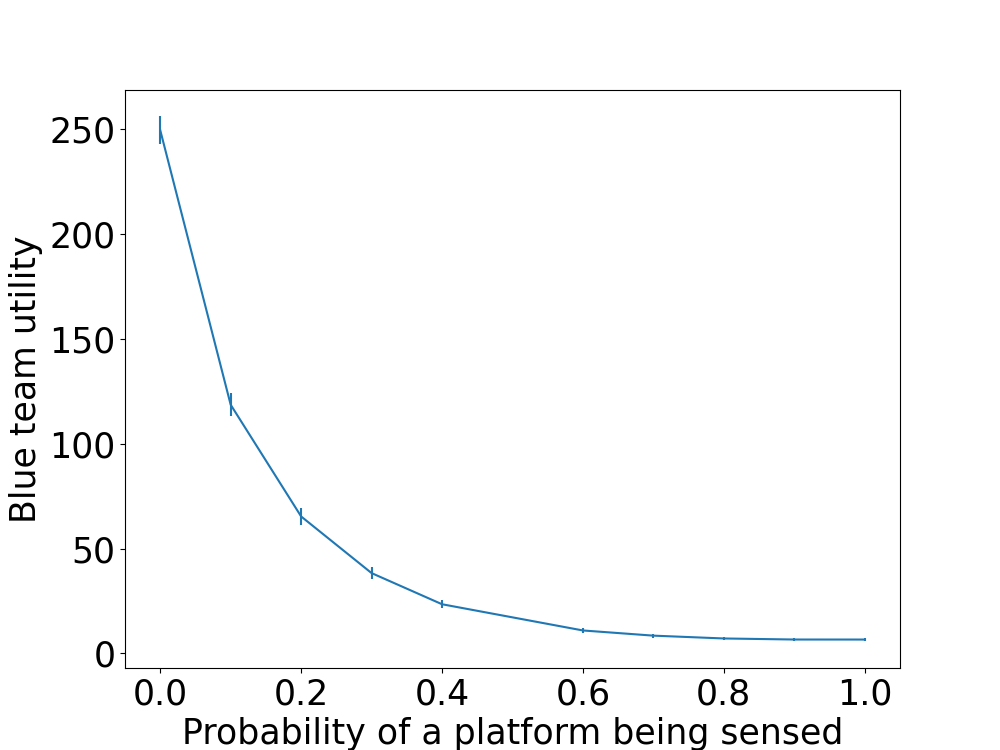}
    \caption{In this experiment, we focus on the {\bfseries\scshape Default}  setting, but changes the probability of each element $D_{i,j} = 1$. We vary the probability from $0$ to $1$ and observe \texttt{Blue}'s utility. We generate 50 random instances with $500$ targets, $10$ sensors, and $\tau=10$ under each setting. Each target's value is drawn within $[0,1]$ uniformly.}
    \label{fig:prob_d}
\end{figure}

In Table \ref{tab:runtime_red_team_ILP}, we show the effectiveness of our ILP solver. Notably, it demonstrates the capability to solve large instances very fast, completing the task within a second.
This proficiency has been valuable in the development of heuristic algorithms for bilevel optimization in the search for identifying the Stackelberg equilibrium. We also show that our ILP solver can solve extensive instances involving hundreds of targets within a single hour in Table \ref{tab:scability_red_team_ILP}. 

\begin{small}
\begin{table}[tbh]
    \centering
    \begin{tabular}{|c|c|c|c|}
    \hline
         \makecell{Utility, \\ \textit{Time (s)}} &  2 & 5 &10\\
         \hline
        5 & \makecell{1.76$\pm$ 0.75, \\ \textit{0.001}}&  \makecell{ 1 $\pm$ 0.63, \\ \textit{0.002 $\pm$ 0.02}} & \makecell{ 0.38 $\pm$ 0.43, \\ \textit{0.003 $\pm$ 0.002}}\\
        \hline
        25  & \makecell{8.73 $\pm$ 1.27, \\ \textit{0.004 $\pm$ 0.002}}& \makecell{ 4.55 $\pm$ 1.3, \\ \textit{0.008 $\pm$ 0.002}} &  \makecell{ 1.39 $\pm$ 0.84, \\ \textit{0.01 $\pm$ 0.003}}\\
        \hline
        75  & \makecell{26$\pm$ 2.8, \\ \textit{0.01 $\pm$ 0.007}}& \makecell{ 13.9$\pm$ 2.33, \\ \textit{0.02 $\pm$ 0.007}}&  \makecell{ 4.41 $\pm$ 1.57, \\ \textit{0.04 $\pm$ 0.007}}\\
        \hline
    \end{tabular}
    \caption{\textsc{Best Red Response} ILP running time: Each row represents the number of targets, and each column represents the number of sensors. $\tau=2$ for every setting. Each element represents the average \texttt{Blue}'s utility (first row) and the average solving time (\textit{italic second row}) of 50 randomly generated instances under {\bfseries\scshape Default} game setting.}
    \label{tab:runtime_red_team_ILP}
\end{table}
\end{small}

\begin{small}
\begin{table}[tbh]
    \centering
    \begin{tabular}{|c|c|c|c|}
    \hline
         \makecell{Utility, \\ \textit{Time (s)}} &  5 & 10 & 20\\
         \hline
        600 & \makecell{175$\pm$ 6, \\ \textit{0.73 $\pm$ 0.01}}&  \makecell{ 77 $\pm$ 5, \\ \textit{1.55 $\pm$ 0.22}}& \makecell{ 4.01 $\pm$ 1.53, \\ \textit{2.79 $\pm$ 0.04}}\\
        \hline
        800  & \makecell{234 $\pm$ 6.5, \\ \textit{1.02 $\pm$ 0.01}}&  \makecell{ 103 $\pm$ 6.1, \\ \textit{2.09 $\pm$ 0.27}}& \makecell{ 5.28 $\pm$ 2.06, \\ \textit{3.82 $\pm$ 0.04}}\\
        \hline
        1000  & \makecell{292$\pm$ 7.2, \\ \textit{1.23 $\pm$ 0.02}}&  \makecell{ 131 $\pm$ 6.7, \\ \textit{2.62 $\pm$ 0.35}}& \makecell{ 6.4$\pm$ 2, \\ \textit{5 $\pm$ 0.06}}\\
        \hline
        5000  & \makecell{1478 $\pm$ 20, \\ \textit{6.5 $\pm$ 0.05}}&  \makecell{ 662 $\pm$ 13.1, \\ \textit{13.66 $\pm$ 1.48}}& \makecell{ 33.1 $\pm$ 3.98, \\ \textit{24.76 $\pm$ 0.051}}\\
        \hline
        10,000  & \makecell{2959$\pm$ 24, \\ \textit{12.9 $\pm$ 0.18}}&  \makecell{ 1321 $\pm$ 20, \\ \textit{29.8 $\pm$ 2.6}}& \makecell{ 66.3$\pm$ 5.57, \\ \textit{62.59 $\pm$ 0.83}}\\
        \hline
    \end{tabular}
    \caption{Each row represents the number of targets, and each column represents the number of sensors. $\tau=10$ for every setting. Each element represents the average \texttt{Blue}'s utility (first row) and the average solving time (\textit{italic second row}) of the ILP for 50 randomly generated instances under {\bfseries\scshape Default} game setting.}
    \label{tab:scability_red_team_ILP}
\end{table}
\end{small}

In the remaining parts of this section, we explore a new game setting ({\bfseries\scshape Append}) for generating ESGs. The new method is similar to {\bfseries\scshape Default}, with the distinction that each element $D_{i,j}= 1$ with a 0.5 probability. In essence, this configuration increases the likelihood of each target being sensed compared to the {\bfseries\scshape Default} setting, thereby resulting in a \textit{stronger} \texttt{Red} sensing model.

\subsection{Computing the Follower Strategy}\label{sub:follower_strat}

Similar to the {\bfseries\scshape Default} game setting, we show the scalability results  of {\bfseries\scshape Append} in table \ref{tab:scability_red_bestreponse_time}.

\begin{small}
\begin{table}[H]
    \centering
    \begin{tabular}{|c|c|c|c|}
    \hline
         \makecell{Utility, \\ \textit{Time (s)}} &  5 & 10 & 20\\
         \hline
        600 & \makecell{67.8 $\pm$ 4.3, \\ \textit{0.74 $\pm$ 0.06}}&  \makecell{ 13.82 $\pm$ 3.96, \\ \textit{4.68 $\pm$ 8.46}}& \makecell{ 0.2 $\pm$ 1.38, \\ \textit{2.85 $\pm$ 0.04}}\\
        \hline
        800  & \makecell{91.1 $\pm$ 5.8, \\ \textit{1.02 $\pm$ 0.1}}&  \makecell{ 19.6 $\pm$ 7.75, \\ \textit{6.27 $\pm$ 10.98}}& \makecell{ 0.13 $\pm$ 0.36, \\ \textit{3.66 $\pm$ 0.04}}\\
        \hline
        1000  & \makecell{114 $\pm$ 7.9, \\ \textit{1.29 $\pm$ 0.14}}&  \makecell{ 24.5 $\pm$ 5.9, \\ \textit{6.06 $\pm$ 4.7}}& \makecell{ 0.25$\pm$ 1.45, \\ \textit{4.77 $\pm$ 0.06}}\\
        \hline
    \end{tabular}
    \caption{Each row represents the number of targets, and each column represents the number of sensors. $\tau=10$ for every setting. Each element represents the average \texttt{Blue}'s utility (first row) and the average solving time (\textit{italic second row}) of the ILP for 50 randomly generated instances under {\bfseries\scshape Append} game setting.}
    \label{tab:scability_red_bestreponse_time}
\end{table}
\end{small}

In this new setting, an intriguing observation emerges as the number of sensors increases significantly: the runtime of our ILP decreases, given that the abundant sensors can effectively sense all targets (e.g., when the number of sensors increased from 10 to 20.).

\subsection{Additional Results from Computing the Stackelberg Equilbrium} 

In this subsection, we begin by showing Figure \ref{fig:mu_hyper_tuning}, illustrating the impact of the choice of ration ($\mu$) on the SA algorithm discussed in Section \ref{sec:Stackle}. 
Specifically, we present and test three quadratic-time greedy heuristics to build the sensing plan $\psi$ iteratively by trying to sense the most valuable targets first. 
We consider the targets in decreasing order of their value. Let $T'$ be the already processed targets and $t_\ell$ the target to consider. Moreover, let  $S'\subseteq S$ the set of sensors $s$ so that $\psi$ remains a valid sensing plan after adding $t_\ell$ to $\psi(s)$, i.e., the sensors that are currently free to sense $t_\ell$. 
If $S'$ is empty, then we do not assign $t_\ell$ to any sensor, implying that it will be won by \texttt{Blue}. 
Otherwise, we apply three different methods to decide which sensor from $S'$ to pick: 
\begin{description}
    \item[random] Randomly select a sensor from $S'$.
    \item[remaining\_value] Pick the sensor $s$ from $S'$ that has the lowest summed value of remaining targets that $s$ is capable of sensing, i.e., $\argmin_{s_j\in S'} \sum_{t_i\in T \setminus T': D_{i,j}=1} v_i$.
    \item[harm] Pick the sensor $s$ from $S'$ that where assigning $t_\ell$ does the least harm: The harm that $t_\ell$ does to $s$ in $\psi$ summed value of remaining targets that $s$ is capable of sensing that it can no longer sense when $t_\ell$ is assigned to $s$, i.e., $\argmin_{s_j\in S'} \sum_{t_i\in T \setminus T': D_{i,j}=1 \text{ and } |i-\ell|\leq \tau } v_i$.
\end{description}

\begin{figure}[tbh]
    \centering
    \includegraphics[width=0.45\textwidth]{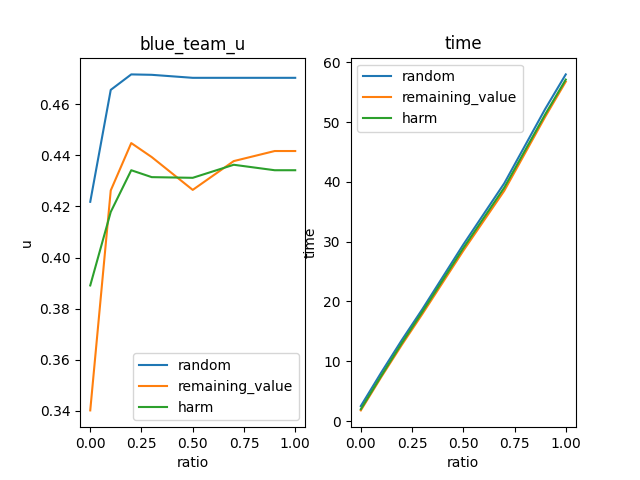}
    \caption{In this experiment, we generate 50 random instances with $10$ targets, $5$ sensors, and $\tau=2$ under {\bfseries\scshape Append} game setting. Each target's value is drawn within $[0,1]$ uniformly.} 
    \label{fig:mu_hyper_tuning}
\end{figure}

We also present the bilevel ILP's scalability results under the new {\bfseries\scshape Append} game setting in Table \ref{tab:append_scability_bilevel_time}.

\begin{small}
\begin{table}[H]
    \centering
    \begin{tabular}{|c|c|c|c|}
    \hline
         \makecell{Utility, \\ \textit{Time (s)}} &  2 & 3 & 5\\
         \hline
        5 & \makecell{0.72 $\pm$ 0.55, \\ \textit{0.73 $\pm$ 0.17}}&  \makecell{ 0.41 $\pm$ 0.49, \\ \textit{0.85 $\pm$ 0.12}}& \makecell{ 0.14$\pm$ 0.33, \\ \textit{1.02 $\pm$ 0.04}}\\
        \hline
        7 & \makecell{0.92 $\pm$ 0.43, \\ \textit{121 $\pm$ 13}}&  \makecell{ 0.57 $\pm$ 0.53, \\ \textit{134 $\pm$ 17}}& \makecell{ 0.18$\pm$ 0.33, \\ \textit{155 $\pm$ 16}}\\
        \hline 
        8 & \makecell{0.86 $\pm$ 0.51, \\ \textit{1841 $\pm$ 356}}&  \makecell{ 0.41 $\pm$ 0.42, \\ \textit{2151 $\pm$ 253}}& \makecell{ 0.13$\pm$ 0.27, \\ \textit{2832 $\pm$ 549}}\\
        \hline 
        9  & \makecell{ n/a, \\ \textit{31358}}&  \makecell{ n/a, \\ \textit{36181}}& \makecell{ n/a, \\ \textit{41999}}\\
        \hline
    \end{tabular}
    \caption{Each row represents the number of targets, and each column represents the number of sensors. $\tau=2$ for every setting. Each element (when $n=5, 7, 8$) represents the average \texttt{Blue}'s utility (first row) and the solving time \textit{(italic second row)} of the bilevel ILP for 50 randomly generated instances under the new {\bfseries\scshape Append} game setting. At $n=9$, the computational time is prohibitively high. Therefore, we conduct a single run on a random instance and record the solving time. Since the utility of this individual instance is not comparable to the average utility derived from 50 random instances, we have omitted it from the table.}
    \label{tab:append_scability_bilevel_time}
\end{table}
\end{small}

We also provide a comparison of heuristic algorithms under the new {\bfseries\scshape Append} game setting in Table \ref{tab:approx_greedy}.

\begin{table}[tbh]
    \centering
    \begin{tabular}{|c|c|c|}
         \hline
         \makecell{Utility, \\ \textit{Time (s)}}  & (7,3,2) & (75, 10, 5) \\
         \hline 
         OPT & \makecell{0.58 $\pm$ 0.53, \\\textit{ 134 $\pm$ 17}} & n/a \\
         \hline
         SA &\makecell{0.58 $\pm$ 0.53, \\ \textit{6.13 $\pm$ 0.87}}  & \makecell{2.26 $\pm$ 1.13, \\\textit{27335 $\pm$ 347}}  \\
         \hline 
         SA\_Relax & \makecell{0.58 $\pm$ 0.53, \\ \textit{0.043 $\pm$ 0.01}} & \makecell{0.56 $\pm$ 0.48, \\ \textit{46.87} $\pm$ 0.92}\\
         \hline
         Random & \makecell{0.50 $\pm$ 0.54, \\ \textit{0.001}} & \makecell{0.15 $\pm$ 0.27, \\ \textit{0.001}} \\
         \hline
    \end{tabular}
    \caption{Compare the approximability of different greedy algorithms in terms of \texttt{Blue}'s utility and solving time. We generate 50 random instances and report the averaged value plus the standard deviation under {\bfseries\scshape Append} game setting.}
    \label{tab:approx_greedy}
\end{table}
\vspace{-5mm}

\subsection{Additional Results from Non-Coordinated Sensing}\label{appendix_sec:nonCoordinated}
In Table \ref{tab:scability_tau3}, we present the scalability results of the ILP that solves for optimal \texttt{Blue} responses in the non-coordinated \texttt{Red} sensing setting under {\bfseries\scshape Append} game setting. 

\begin{small}
\begin{table}[H]
    \centering
    \begin{tabular}{|c|c|c|c|}
    \hline
         \makecell{Utility, \\ \textit{Time (s)}} &  2 & 3 & 5\\
         \hline
        5 & \makecell{1.26 $\pm$ 0.46, \\ \textit{0.01 $\pm$ 0.04}}&  \makecell{ 0.96 $\pm$ 0.47, \\ \textit{0.02 $\pm$ 0.04}}& \makecell{ 0.52$\pm$ 0.42, \\ \textit{0.02 $\pm$ 0.04}}\\
        \hline
        10 & \makecell{2.34 $\pm$ 0.7, \\ \textit{4.07 $\pm$ 16.8}}&  \makecell{1.85 $\pm$ 0.89, \\ \textit{90.2 $\pm$ 294.6}}& \makecell{ 1$\pm$ 0.73, \\ \textit{293 $\pm$ 565}}\\
        \hline 
        15  & \makecell{ 3.79 $\pm$ 0.87, \\ \textit{482 $\pm$ 2551}}&  \makecell{ 3.01 $\pm$ 1.16, \\ \textit{50.4 $\pm$ 18.9}}& \makecell{ n/a, \\ \textit{40296} $\pm$ n/a}\\
        \hline
        20  & \makecell{ 4.61 $\pm$ 1, \\ \textit{889 $\pm$ 1530}}&  \makecell{ n/a, \\ \textit{9380 $\pm$ n/a}}& \makecell{ n/a, \\ \textit{n/a}}\\
        \hline
    \end{tabular}
    \caption{Each row represents the number of targets, and each column represents the number of sensors. $\tau=2$ for every setting. Each element represents the average \texttt{Blue}'s utility (first row) and the solving time \textit{(italic second row)} of \texttt{Blue}'s best response ILP for 50 randomly generated instances under {\bfseries\scshape Append} game setting. For settings where the computational time is prohibitively high, we conduct a single run on a random instance and record the single solving time. Thus, their standard deviation value is recorded as ``n/a''.  Additionally, the utility of this single instance is not comparable to the average utility derived from 50 random instances, we have omitted it from the table.}
    \label{tab:scability_tau3}
\end{table}
\end{small}

For the non-coordinated sensing setting, we also provide a comparison of heuristic algorithms under the new scenario in Table \ref{tab:approx_greedy2}.

\begin{table}[tbh]
    \centering
    \begin{tabular}{|c|c|c|}
         \hline
         \makecell{Utility, \\ \textit{Time (s)}}  & (10,5,2) & (75, 10, 5) \\
         \hline 
         OPT &  \makecell{ \textbf{1}$\pm$ 0.73, \\ \textit{293 $\pm$ 565}} & n/a \\
         \hline
         SA &\makecell{0.83 $\pm$ 0.67, \\ \textit{0.76 $\pm$ 0.02}}  & \makecell{4.57 $\pm$ 1.46, \\ \textit{523 $\pm$ 6.28}}   \\
         \hline 
         SA\_Relax & \makecell{0.94 $\pm$ 0.71, \\ \textit{0.02 $\pm$ 0.003}} & \makecell{3.47 $\pm$ 1.02, \\ \textit{3.13 $\pm$ 0.12}}\\
         \hline
         Random & \makecell{0.38 $\pm$ 0.5, \\ \textit{0.001}} & \makecell{1.28 $\pm$ 1.05, \\ \textit{0.001}} \\
         \hline
    \end{tabular}
    \caption{Compare the approximability of different greedy algorithms in terms of \texttt{Blue}'s utility and solving time. We generate 50 random instances and report the averaged value plus the standard deviation under {\bfseries\scshape Append} game setting.}
    \label{tab:approx_greedy2}
\end{table}

\subsubsection{Power of Coordination}
Finally, in Table \ref{tab:power_coor_appendix0}, we present the power of coordination results under {\bfseries\scshape Append} game setting with \texttt{Red} that has stronger sensing capabilities.  
\begin{table}[tbh]
    \centering
    \begin{tabular}{|c|c|c|}
    \hline
         &  Greedy  & Coordination\\
         \hline
        (5, 2, 2) & 1.26 $\pm$ 0.46&  0.72 $\pm$ 0.55  \\
        \hline
        (5, 3, 2) &  0.96 $\pm$ 0.47 &   0.41 $\pm$ 0.49 \\
        \hline 
        (5, 5, 2)  & 0.52 $\pm$ 0.42& 0.14 $\pm$ 0.33 \\
        \hline
    \end{tabular}
    \caption{{\bfseries\scshape Append} setting: Each element represents the average \texttt{Blue}'s utility for 50 randomly generated instances. The game is a constant-sum game. Therefore, the decrease in \texttt{Blue}'s utility corresponds to an increase in \texttt{Red}'s utility, showing the power of coordination.}
    \label{tab:power_coor_appendix0}
\end{table}

Moreover, in scenarios with large instance sizes, such as $n=75, k=10, \tau=5$, where the optimal (bilevel) ILP is unsolvable, we can compare \texttt{Blue}'s approximately optimal utility under the best heuristic algorithms. Specifically, as shown in Table \ref{tab:approx_greedy}, the average \texttt{Blue}'s utility under the SA algorithm is $2.26 \pm 1.13$ for 9 instances. Given the same 9 instances, the average \texttt{Blue}'s utility when escaping from \textit{non-coordinated sensing} is $4.43 \pm 1.41$, which is approximately \textit{twice} the value observed in the \textit{coordinated sensing setting}. Once again, due to the constant-sum nature of this game, the utility loss for \texttt{Blue} in transitioning from non-coordinated sensing to coordinated sensing essentially represents the utility gain for \texttt{Red}, highlighting the power of coordination.

\end{document}